\documentclass[aps,superscriptaddress,twocolumn,twoside,floatfix,pra,nofootinbib,a4paper]{revtex4-2}
\usepackage{bm,color,bbm}
\usepackage{hyperref, mathtools,graphicx,soul}

\usepackage{subfigure}

\newcommand{\beq}{\begin{equation}}
\newcommand{\eeq}{\end{equation}}
\newcommand{\bqa}{\begin{eqnarray}}
\newcommand{\eqa}{\end{eqnarray}}
\newcommand{\nn}{\nonumber}

\newcommand{\bra}[1]{ \langle{#1} |}
\newcommand{\ket}[1]{ |{#1} \rangle}

\newcommand{\sq}[1]{\left[ {#1} \right]}

\newcommand{\tr}[1]{{\rm Tr}\sq{ {#1} }}

\newcommand{\mf}{\mathbf}

\definecolor{maroon}{rgb}{0.7,0,0}

\definecolor{ngreen}{rgb}{0.3,0.7,0.3}

\definecolor{golden}{rgb}{0.8,0.6,0.1}

\definecolor{npurple}{rgb}{0.3,0,0.6}

\usepackage{amsmath,amssymb,amsthm}
\usepackage{amsmath,kbordermatrix,blkarray}
\usepackage{braket}
\usepackage{color}
\usepackage{mathtools}
\usepackage{natbib}
\usepackage{slashbox}

\allowdisplaybreaks

\newcommand{\ketbra}[2]{|#2\rangle\langle#1|}

\newcommand{\be}{\begin{equation}}
\newcommand{\ee}{\end{equation}}
\newcommand{\ba}{\begin{eqnarray}}
\newcommand{\ea}{\end{eqnarray}}

\newcommand{\one}{\bf{1}}

\newcommand{\etal}{{\it{et. al. }}}

\newtheorem{observation}{Observation}
\newtheorem{definition}{Definition}
\newtheorem{thm}{Theorem}
\newtheorem{cor}{Corollary}
\newtheorem{lem}{Lemma}

\newtheorem{prop}{Proposition}

\newtheorem{result}{Result}
\begin{document}

\title{Operational  simultaneous correlations in complementary bases of bipartite states via one-sided semi-device-independent steering}  

\author{Chellasamy Jebarathinam}

 \affiliation{Department of Physics and Center for Quantum Information Science, National Cheng Kung University, Tainan 701, Taiwan (R.O.C.)}

  \affiliation{Physics Division, National Center for Theoretical Sciences, National Taiwan University, Taipei 106319, Taiwan (R.O.C.)}

 \author{Debarshi Das}
 \email{dasdebarshi90@gmail.com}
\affiliation{Department of Physics and Astronomy, University College London, Gower Street, WC1E 6BT London, England, United Kingdom} 

\author{Huan-Yu Ku}
\affiliation{Department of Physics, National Taiwan Normal University, Taipei 116059, Taiwan (R.O.C.)}

\author{Debasis Sarkar}
 \email{dsappmath@caluniv.ac.in}
\address{Department of Applied Mathematics, University of Calcutta, 92, A.P.C. Road, Kolkata, 700009, India}

\author{Hsi-Sheng Goan}

 \email{goan@phys.ntu.edu.tw}
 \affiliation{Department of Physics and Center for Theoretical Physics,
National Taiwan University, Taipei 106319, Taiwan (R.O.C.)}
\affiliation{Center for Quantum Science and Engineering, National Taiwan University, Taipei 106319, Taiwan (R.O.C.)}
 \affiliation{Physics Division, National Center for Theoretical Sciences, National Taiwan University, Taipei 106319, Taiwan (R.O.C.)}

\begin{abstract}
 Recently, a different form of quantum steering, i.e., certification of quantum steering in a one-sided semi-device-independent way, has been formulated [Jebarathinam \etal Phys. Rev. A 108, 042211 (2023)]. In this work, we use this phenomenon to provide operational simultaneous correlations in mutually unbiased bases as quantified by the measures in [Wu \etal Scientific Reports 4, 4036
(2014)]. 
First, we show that for any bipartite state, such a measure of simultaneous correlations in two mutually unbiased bases can be operationally identified as exhibiting one-sided semi-device-independent steering in a two-setting scenario.
 Next, we demonstrate that for two-qubit Bell-diagonal states, quantifying one-sided semi-device-independent steerability provides an operational quantification of information-theoretic quantification of simultaneous correlations in mutually unbiased bases. Then,  we provide a different classification of two-qubit separable states with the above-mentioned information-theoretic quantification of simultaneous correlations in mutually unbiased bases and the quantification of one-sided semi-device-independent steerability. Finally, we invoke quantum steering ellipsoid formalism to shed intuitions on the operational characterization of simultaneous correlations in complementary bases of two-qubit states via one-sided semi-device-independent steerability. This provides us with a geometric visualization of the results. 
\end{abstract}

\maketitle 
   \date{\today}
   
\section{Introduction}
Measurements in mutually unbiased bases (MUBs) or complementary bases  have been invoked as fundamental resources for optimally observing or quantifying quantum nonlocality \cite{TFR+21, SNC14}.
Quantum steering is a form of nonlocal correlation introduced by  Schr\"{o}dinger \cite{Sch35}. Quantum steering can be demonstrated by invoking the presence of simultaneous correlations in complementary bases of the given composite system \cite{JKK+18}. 
Quantum steering has been formulated as a one-sided device-independent (1SDI) certification of relevant resources (incompatible measurements and entanglement in the state) used to demonstrate the phenomenon \cite{WJD07,UBG+15,KCB+18,DDJ+18,KHCC+22,Wang2024,LLL+24,KHB24}. Quantum information protocols in 1SDI scenario rely on such certification \cite{BCW+12, PCS+15, GA15, LJW+21,SBJ+23,Wang2024,HKB23,LLL+24}. 

Recently, it has been formulated that quantum steering can also be observed with a different level of certification, i.e., in a one-sided semi-device-independent (1SSDI) way, that further  constrains the dimension of the devices at the untrusted side \cite{JDS_PRA23}. Such a different form of steering occurs in states that are indeed unsteerable in a 1SDI scenario. In fact, these states have superunsteerability as formulated in Refs. \cite{DBD+18,DJB+18}. 
Interestingly, superunsteerability also occurs for certain separable states with a nonvanishing quantum discord \cite{OZ01,HV01} as well as global coherence \cite{JKC+24}. 
 Superunsteerability detects the 1SSDI steerable resources, which are unsteerable boxes in 1SDI scenario. Although any unsteerable box is not a resource in the resource theory of 1SDI steering \cite{GA15, DDJ+18}, superunsteerability is a resource in the context of 1SSDI steering due to constraining the dimension at the untrusted side.
Superunsteerability has been proved to be  useful for quantum information processing \cite{JDK+19, JD23}.
In Ref. \cite{JDK+19}, the quantification of superunsteerability, called Schr\"{o}dinger strength,  was introduced. Moreover, in
 Ref. \cite{JKC+24} a resource theory of 1SSDI steering was formulated. Utilizing this resource theory, for the quantum communication task of remote state preparation using separable states \cite{DLM+12}, the precise resource underpinning this task has been identified to be the 1SSDI steerability  with a nonvanishing Schr\"{o}dinger strength.

Apart from steering, in Ref. \cite{WMC+14}, information-theoretic quantities have been introduced to quantify quantum correlation in MUBs directly. Here, MUBs are restricted to those that optimize classical correlation in the definition of quantum discord. The above information-theoretic measures can also reveal simultaneous correlations in mutually unbiased bases (SCMUB) in states that do not violate a Bell inequality or steering inequality but have a nonvanishing quantum discord. 

The present work is motivated by whether the presence of SCMUB in discord can be identified by an operational phenomenon.  
To this end, we provide the operational characterization of SCMUB  via 1SSDI steering. We demonstrate that simultaneous correlations in two MUBs for any bipartite state can be operationally identified as exhibiting 1SSDI steering in a two-setting scenario. To demonstrate the above results, we specialize in bipartite two-qubit states as explicit examples and obtain the following results. 
For a broad class of two-qubit Bell-diagonal states,  we study the relationships between the quantification of 1SSDI steerability, based on Schr\"{o}dinger strength introduced in Ref. \cite{JDK+19} with two (three) setting scenario, and information-theoretic measures of SCMUB  \cite{WMC+14}. Such relationships reveal that quantifying 1SSDI steerability of Bell-diagonal states provides an operational quantification of simultaneous correlations in two and three MUBs, respectively   \cite{WMC+14}. We then provide a different classification of two-qubit discordant states in terms of simultaneous correlations in two and three MUBs and 1SSDI steerability.
Finally, we invoke quantum steering ellipsoid formulation to shed intuitions on the operational SCMUB of two-qubit states via 1SSDI steerability. This gives us a geometric visualization of our results.

\section{Preliminaries}

\subsection{Quantum discord and bipartite coherence}
Quantum discord \cite{HV01, OZ01} is a measure of the quantumness of a bipartite state $\rho_{AB} \in \mathcal{B}(\mathcal{H}_A \otimes \mathcal{H}_B)$. Here $\mathcal{B}(\mathcal{H}_A \otimes \mathcal{H}_B)$ stands for the set of all bounded linear operators acting on the Hilbert space $\mathcal{H}_A \otimes \mathcal{H}_B$.
 Before we introduce the definition of quantum discord, we present the classical correlation used in the definition of quantum discord. 
For a given basis $\{\Pi_{i}^{A}\}$ measured on Alice's subsystem of a bipartite state $\rho_{AB}$,  the Holevo quantity 
$\chi(\rho_{AB}|\{\Pi_{i}^{A}\})$ of the ensemble $\{p_{i},\rho_{B|i}\}$
prepared on Bob's side is given by 
\begin{align}
\chi(\rho_{AB}|\{\Pi_{i}^{A}\}) & =S\big(\rho_B\big)- \sum_{i}p_{i}S\big(\rho_{B|i}\big),
\end{align}
which gives an upper bound on Bob's accessible information about Alice's results. \(S(\sigma) = -\mbox{tr}(\sigma \log_2 \sigma)\) is the von Neumann entropy of a density matrix \(\sigma\).
\begin{definition}
   The classical correlation $C^{\rightarrow}(\rho_{AB})$ is defined by  the Holevo quantity $\chi(\rho_{AB}|\{\Pi_{i}^{A}\})$ optimized over all projective measurements,
i.e.,
\begin{equation}\label{s4m}
C^{\rightarrow}(\rho_{AB}):=\max_{\{\Pi_{i}^{A}\}}\chi(\rho_{AB}|\{\Pi_{i}^{A}\}).
\end{equation} 
\end{definition}

Using the classical correlation of a bipartite state defined above, quantum discord  \cite{HV01, OZ01} from Alice to Bob $D^{\rightarrow}(\rho_{AB})$  can be defined as follows.
\begin{definition}
Quantum discord is defined as
    \begin{align}\label{QDdef}
D^{\rightarrow}(\rho_{AB}):= I(\rho_{AB}) - C^{\rightarrow}(\rho_{AB}).
\end{align}            
Here $I(\rho_{AB})= 
 S(\rho_B) - \tilde{S}(\rho_{B|A})$,
is the quantum mutual information and can be interpreted as the total correlations in \(\rho_{AB}\).  
 \(\tilde{S}(\rho_{B|A}) = S(\rho_{AB}) - S(\rho_A)\) is the ``unmeasured'' quantum conditional entropy \cite{qmi} (see also \cite{Cerf, SN96,GROIS}).
\end{definition}

Quantum discord $D^{\rightarrow}(\rho_{AB})$ as defined above captures the quantumness of correlations in the state from Alice to Bob. Similarly, the quantum discord from Bob to Alice $D^{\leftarrow}(\rho_{AB})$ can be defined.  
$D^{\rightarrow}(\rho_{AB})$  vanishes for a given $\rho_{AB}$ if and only if it is a classical-quantum (CQ) state of the form,
\begin{align}
\rho_{\texttt{CQ}}=\sum_i p_i \ket{i} \bra{i}_A \otimes \rho^{(i)}_B,
\end{align}
where $p_i\ge 0$ and $\sum_i p_i=1$, $\{\ket{i}\}$ forms an orthonormal basis on Alice's Hilbert space and $\rho^{(i)}_B$ are any quantum states on Bob's Hilbert space.

Quantumness, as captured by discord, implies another notion of quantumness called bipartite coherence.
To define incoherent (IC) bipartite states, let us first fix the product basis $\{\ket{i}^A \otimes \ket{j}^B\}$ as the reference basis. 
\begin{definition} \label{bc}
With respect to the reference basis $\{\ket{i}^A \otimes \ket{j}^B\}$, the bipartite IC states  are defined as those whose density matrix is diagonal in such a basis \cite{BCA15,SSD+15}: 
\be \label{IC}
\rho^{IC}_{AB}=\sum_k p_k \sigma^A_k \otimes \tau^B_k, 
\ee
where $p_k$ are probabilities; $\sigma^A_k$ and $\tau^B_k$ are incoherent states on subsystems $A$ and $B$ with respect to the bases  $\{\ket{i}^A \}$ and $\{\ket{j}^B\}$ respectively,
 i.e., $\sigma^A_k=\sum_i p'_{ik} \ketbra{i}{i}$ and $\tau^B_k=\sum_j p''_{jk} \ketbra{j}{j}$ with $p'_{ik}$ and $p''_{jk}$ being the probabilities. Otherwise, it is said to have \textit{bipartite coherence}.  
 \end{definition}
 Biparite coherence can be quantified using any measure of coherence defined for single systems \cite{BCP14,SAP17}.

 Bipartite coherence is present in all discordant states for which one of the discords $D^{\rightarrow}(\rho_{AB})$ and $D^{\leftarrow}(\rho_{AB})$ is nonzero, or both discords are nonzero. In addition, certain product states, which are not discordant, also display bipartite coherence with respect to some fixed product basis.  This implies that quantum discord and quantum coherence are, in general, inequivalent. In fact, the relative entropy of quantum discord, which is nonvanishing for all discordant states,  is upper bounded by the relative entropy of quantum coherence \cite{YXL+15}. However, the basis-free coherence of correlations, defined as the relative entropy of coherence minimized over all local unitary transformations, is exactly equal to the relative entropy of quantum discord \cite{YXL+15}.

 To complete the relationship between coherence and discord, we review some more notions of coherence of correlations. In Refs. \cite{BKP+16, WYY+17, GG17, LZZ+24}, different measures of coherence of correlations in bipartite states have been studied.  In this context, let us note that correlated coherence of bipartite states  in a basis-dependent way \cite{BKP+16,WYY+17} was defined as the difference of coherence in the global state and the total  coherence in the subsystems. Such measures of coherence of correlations vanish on all product states, and moving to the nonproduct states, they are, however, not equivalent to any of the discords. See Ref. \cite{WYY+17} for the analysis of states with vanishing correlated coherence.

In Ref. \cite{JKC+24}, a different notion of coherence of correlations, called global coherence, was introduced. All product states with bipartite coherence have coherence locally. Moving to the nonproduct states,  consider all discordant states which can be created locally from a state with vanishing discord by a local operation, $\Phi_A \otimes \Phi_B$, where $\Phi$ is a completely positive trace-preserving (CPTP) map. All such discordant states have been identified in Ref. \cite{BPP15} in the context of local creation of discord \cite{SKB11,GLB+12, Gio13}. There are also nonproduct  states which are not discordant, but have bipartite coherence due to local coherence in a subsystem or both the subsystems \cite{YMG+16, WYY+17}. All bipartite coherent states whose coherence is due to local coherence in a subsystem or both the subsystems, or can be created locally due to the local creation of discord, are defined to have bipartite coherence locally \cite{JKC+24}. Thus, global coherence of bipartite states can be defined as follows.
\begin{definition}
A bipartite state in $\mathbb{C}^{d_A}\otimes\mathbb{C}^{d_B}$ has global coherence if and only if it has nonzero discord and its discord cannot be generated by a local operation, $\Phi_A \otimes \Phi_B$, from a state with vanishing discord both ways.
\end{definition}
According to this definition, all entangled states have global coherence because they are not separable to have local generation of bipartite coherence. Note that global coherence can also be present in separable states.

To identify separable states having global coherence,
 let us introduce the following decomposition of bipartite states \cite{DVB10}.
 Given a bipartite state $\rho_{AB}$, it can be decomposed as a sum of arbitrary bases of Hermitian operators $\{A_i\}$ and $\{B_i\}$ as 
\be
\rho_{AB} =\sum^{d^2_A}_{n=1}
\sum^{d^2_B}_{m=1} r_{nm} A_n \otimes B_m, 
\ee
where $d_A$ ($d_B$) is the dimension of the Hilbert space $A$
($B$). Using this representation, correlation matrix $R = (r_{nm})$ is defined, which can be rewritten using its singular value decomposition and cast in a diagonal representation as 
\begin{equation}\label{LR}
\rho_{AB}=\sum^{L_R}_{n=1} c_n S_n \otimes F_n. 
\end{equation}
Here, $L_R$ is the rank of $R$ and quantifies how many product operators are needed to represent $\rho_{AB}$.
The value of $L_R$ can be used to witness the presence of nonzero quantum discord \cite{DVB10}. For states with zero discord one-way or both ways, $L_R$ is bounded from above by the minimum dimension of the subsystems
$d_{\min} = \min\{d_A, d_B \}$. On the other hand, in general,
 the correlation rank is bounded by the square of $d_{\min}$:
$L_R \le d^2_{\min}$. Therefore, states with $L_R > d_{\min}$ will be necessarily discordant. However, as shown in Ref. \cite{BPP15}, for quantum states whose discord can be created by a local operation, $\Phi_A \otimes \Phi_B$,  $L_R \le d_{\min}$. 
We can now state the following observation. 
  \begin{observation}
     A separable state in $\mathbb{C}^{d_A}\otimes\mathbb{C}^{d_B}$ has global coherence if and only if $L_R$ in Eq. (\ref{LR}) satisfies $L_R > d_{\min}$.
 \end{observation}
 The above observation follows that for any bipartite coherent state with $L_R \le  d_{\min}$, its bipartite coherence is present locally in a subsystem or both subsystems.  See Fig. \ref{Fig:Coh} for a depiction of the relations between states with no global coherence and other states that are product states, states with vanishing asymmetric correlated coherence (Eq. (8) in Ref. \cite{LZZ+24}), denoted by  $C^{\rm {asym}}_{\rm {corr}}$, states with vanishing basis-free coherence \cite{YXL+15}, denoted by $C_{\rm {BF}}$, and all bipartite states. These relations imply that global coherence is stronger than all notions of coherence of correlations \cite{BKP+16, WYY+17, GG17, LZZ+24}.

\begin{figure}[t!]
\begin{center}
\includegraphics[width=8cm]{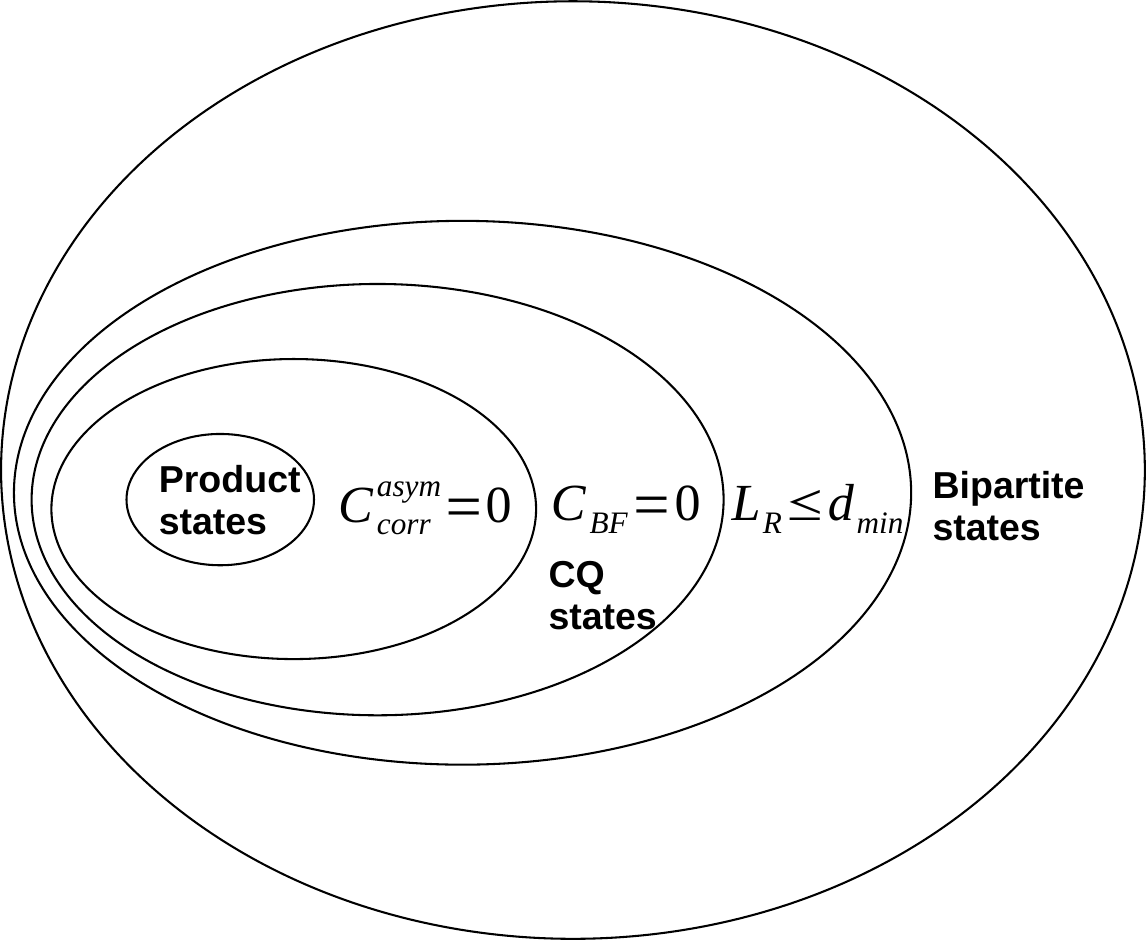}
\end{center}
\caption{Schematic illustration of the connections of states with no global coherence with other states:$ \{\rm {Product} \hspace{0.1cm} \rm{states}\}  \subset \{C^{\rm {asym}}_{\rm {corr}}=0 \} \subset \{C_{\rm {BF}}=0\} \subset \{ L_R \le d_{\min} \} \subset \{\rm {Bipartite} \hspace{0.1cm} \rm{states}\} $. Here, asymmetric correlated coherence studied in Ref. \cite{LZZ+24} and  basis-free coherence  studied in Ref. \cite{YXL+15}  are denoted by $C^{\rm {asym}}_{\rm {corr}}$ and $C_{\rm {BF}}$, respectively.  
\label{Fig:Coh}}
\end{figure}

\subsection{Information-theoretic quantification of correlations in complementary bases} 
We present the series of measures proposed in Ref. \cite{WMC+14}  that seek to capture the genuine quantum correlations in discord via SCMUB.
 For any given bipartite state, there may exist more than one basis that gives rise to  $C^{\rightarrow}(\rho_{AB})$ in Eq. (\ref{s4m}).
\begin{definition}
$C^{\rightarrow}_1$-basis is defined as the basis $\left\{ \left|\mathcal{A}_{i}^{1}\right\rangle _{A}|i=1,\cdots,d_{A}\right\} $ that optimizes $\chi(\rho_{AB},\{\Pi_{i}^{A}\})$ in 
the definition of classical correlation in Eq. (\ref{s4m}). 
\end{definition}
To quantify the simultaneous correlation in two MUBs, $Q^\rightarrow_{2}(\rho_{AB})$, a second basis $\left\{ \left|a_{j}^{2}\right\rangle _{A}|j=1,\cdots,d_{A}\right\} $,
which is mutually unbiased to the $C^{\rightarrow}_1$-basis $\left\{ \left|\mathcal{A}_{i}^{1}\right\rangle _{A}|i=1,\cdots,d_{A}\right\} $, was considered. These are MUBs in the sense that $\left|\left\langle \mathcal{A}_{i}^{1}|a_{j}^{2}\right\rangle \right|=\frac{1}{\sqrt{d_{A}}}$. 
\begin{definition}
Quantum correlation, $Q^\rightarrow_{2}(\rho_{AB})$, to capture the simultaneous correlation in two MUBs is defined as the residual correlation in the second basis mutually unbiased to the $C^{\rightarrow}_1$-basis, given by
\begin{equation}
Q^{\rightarrow}_{2}(\rho_{AB}) \equiv \max_{\{\left|\mathcal{A}_{i}^{1}\right\rangle _{A}\}} \max_{\{\left|a_{j}^{2}\right\rangle _{A}\}}\chi\left(\rho_{AB}|\left\{\left|a_{j}^{2}\right\rangle _{A} \left\langle a_{j}^{2}\right| \right\}\right) . \label{defQ2}
\end{equation}
\end{definition}
If there is only one $C^{\rightarrow}_1$-basis, the second maximization over the $C^{\rightarrow}_1$-bases $\left\{ \left|\mathcal{A}_{i}^{1}\right\rangle _{A}\right\} $ in (\ref{defQ2}) is not necessary.
If there is more than one $C^{\rightarrow}_1$-basis, and not all of them achieve the maximum in (\ref{defQ2}), then the $C^{\rightarrow}_1$-bases are redefined as those that also reach the maximum in (\ref{defQ2}).

Before we define quantum correlation, $Q^\rightarrow_{3}(\rho_{AB})$, to capture the simultaneous correlation in three MUBs, we introduce the following definition.
\begin{definition}
A basis $\{\left|a_{j}^{2}\right\rangle _{A}\}$ that achieves the maximum quantum correlation $Q^{\rightarrow}_2$ in (\ref{defQ2}) is called a $Q^{\rightarrow}_2$-basis, and is denoted as
$\left\{ \left|\mathcal{A}_{j}^{2}\right\rangle _{A}|j=1,\cdots,d_{A}\right\} $.
\end{definition}
$Q^\rightarrow_{3}(\rho_{AB})$ is defined similar to $Q^\rightarrow_{2}(\rho_{AB})$  as follows.
\begin{definition}
Quantification of the simultaneous correlation in three MUBs, $Q^\rightarrow_{3}(\rho_{AB})$, is defined as the residual correlation in a third mutually unbiased basis  as
\begin{equation}
Q^{\rightarrow}_{3}(\rho_{AB}) \equiv \max_{\{\left|\mathcal{A}_{i}^{1}\right\rangle _{A}\}} \max_{\{\left|\mathcal{A}_{j}^{2}\right\rangle _{A}\}}  \max_{\{\left|a_{k}^{3}\right\rangle _{A}\}}\chi\left(\rho_{AB}|\left\{\left|a_{k}^{3}\right\rangle _{A} \left\langle a_{k}^{3}\right| \right\}\right),
\label{defQ3}
\end{equation}
where $\left\{ \left|a_{k}^{3}\right\rangle _{A}|k=1,\cdots,d_{A}\right\} $
is any basis mutually unbiased to both $\left\{ \left|\mathcal{A}_{i}^{1}\right\rangle _{A}\right\} $
and $\left\{ \left|\mathcal{A}_{j}^{2}\right\rangle _{A}\right\} $.
\end{definition}
An optimum basis $\{\left|a_{k}^{3}\right\rangle _{A}\}$ to achieve the maximum in (\ref{defQ3}) is called a $Q^\rightarrow_3$-basis.

For a given bipartite state, the authors of Ref. \cite{WMC+14} have thus proposed the series of measures $\{C_1^{\rightarrow},Q^{\rightarrow}_2, Q^{\rightarrow}_3,\cdots, Q^{\rightarrow}_{d_A+1}\}$, here $C_1^{\rightarrow}$ denotes classical correlation in Eq. (\ref{s4m}), provided that there exists at least $(d_A+1)$ MUBs for subsystem $A$, to characterize correlations in the bipartite state. In the case that Alice's system is qubit, correlations of the state are described with the set $\{C_1^{\rightarrow},Q^{\rightarrow}_2, Q^{\rightarrow}_3\}$ as there are at most three MUBs for a qubit.
A nonzero value of   $Q^{\rightarrow}_2$ requires that the bipartite state has global coherence as we show later, in addition to a nonvanishing quantum discord $D^{\rightarrow}$. See Fig. \ref{Fig:Hierarchy} for the hierarchy of quantumness in bipartite states as  captured by discord, global coherence, and the simultaneous correlations in two and three MUBs.

\begin{figure}[t!]
\begin{center}
\includegraphics[width=8cm]{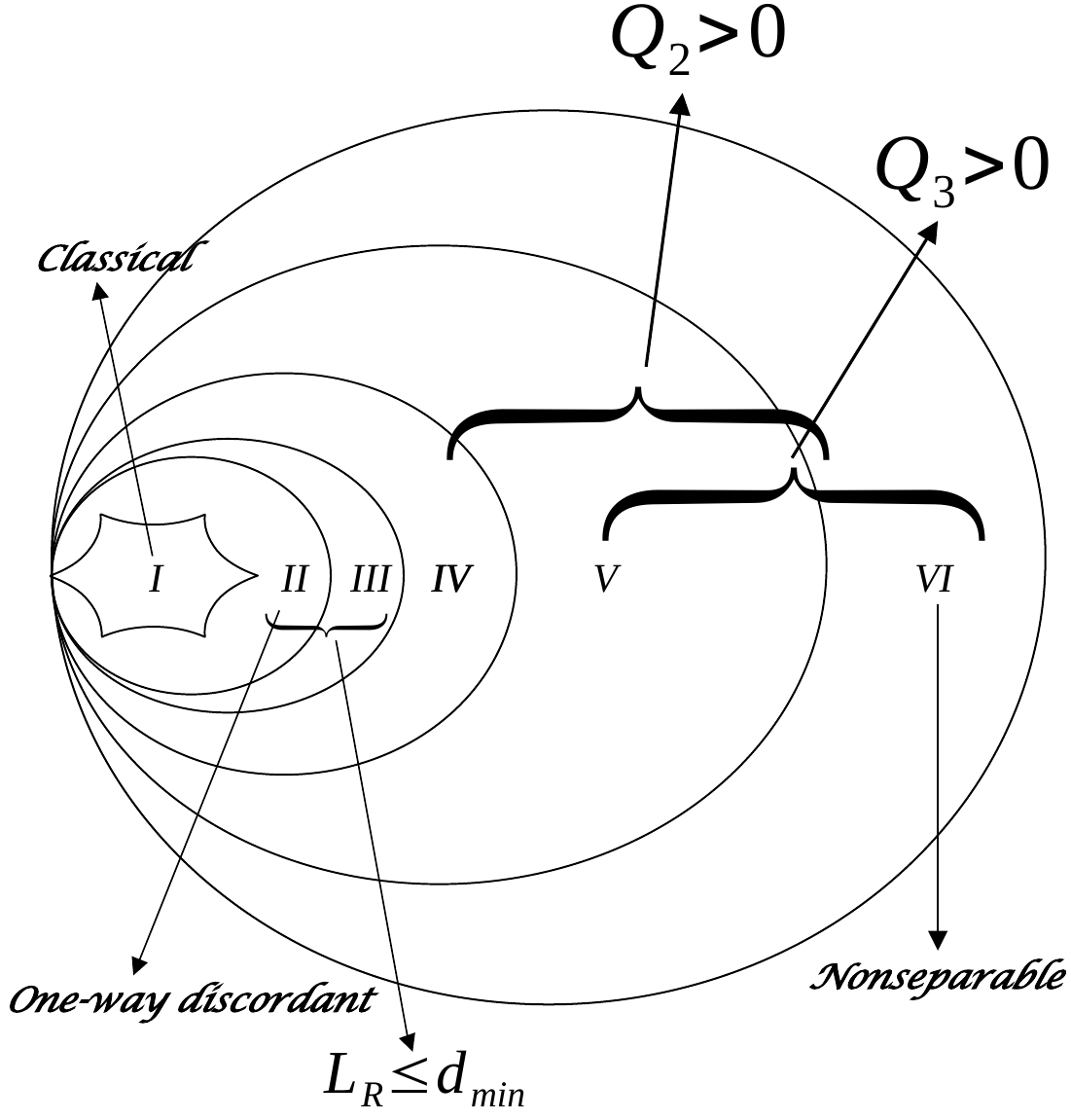}
\end{center}
\caption{Hierarchy of correlations in bipartite quantum states. The region $I$ represents correlations in classical states which have zero discord both ways. 
The regions $II$ and $III$ represent correlations in one-way and two-way discordant states, respectively, which have no global coherence (i.e., $L_R \le d_{\min}$). The regions $IV-VI$ represent simultaneous correlations in two MUBs in discord. The regions $V$ and $VI$ represent the simultaneous correlations in three MUBs in discord.  
\label{Fig:Hierarchy}}
\end{figure}

We present the values of quantifiers in $\{C_1^{\rightarrow},Q^{\rightarrow}_2, Q^{\rightarrow}_3\}$ for the Bell-diagonal states, which admit the form,
\be
\tau=\frac{1}{4}\left(\openone \otimes \openone +\sum_i c_i \sigma_i \otimes \sigma_i \right)=\sum_{ab} \lambda_{ab} \ketbra{\beta_{ab}}{\beta_{ab}}, \label{BD}
\ee
where $\sigma_i$, $i\in \{1,2,3\}$, are the Pauli matrices,  $\lambda_{ab}$ (with $a,b\in \{0,1\}$) denote the eigenvalues of the Bell-diagonal states which are given by 
$\lambda_{ab}=\frac{1}{4}\left[1+(-1)^a c_1 - (-1)^{a+b} c_2+(-1)^b c_3 \right]$
and $\ket{\beta_{ab}}=\frac{1}{\sqrt{2}}(\ket{0,b}+(-1)^a \ket{1,1 \oplus b}$
are the Bell states.  Here, we take $|c_1| \ge |c_2| \ge |c_3|$ with local unitary transformations. For this family of states, the eigenvalues satisfy the ordering $\lambda_{01} \ge  \lambda_{00} \ge \lambda_{10} \ge \lambda_{11}$. This family of states are entangled if and only if $\lambda_{01} > 1/2$. The quantum discord of this family of states $D^\rightarrow(\tau)>0$ if and only if 
$|c_2|>0$ \cite{DVB10}.

The optimal decomposition of the reduced state $\rho_B$ that achieves the  
classical correlation $C^{\rightarrow}(\tau)$ \cite{SSJ+12} is given by
\begin{align}
\rho_B&=\frac{1}{2} \rho_{B|0}+\frac{1}{2} \rho_{B|1},
\end{align}
where, 
\begin{align}
 \rho_{B|0/1}=\frac{1}{2} \left( \openone \pm c_1 \sigma_1 \right).
\end{align}
It has been checked that the ensemble corresponding to this decomposition can be achieved by Alice performing measurement in the basis of $\sigma_1$.
We also note that this is the only basis that achieves $C_1^{\rightarrow}(\tau)$ given by 
\begin{equation}\label{C1BD}
C_1^{\rightarrow}(\tau)=1-h\left(\frac{1+|c_1|}{2}\right),
\end{equation}
where $h(\cdot)$ is the binary entropic function.
From this fact, it follows that a second basis, which is mutually unbiased to the basis of $\sigma_1$, gives rise to $Q^{\rightarrow}_2(\tau)$. The second basis which is mutually unbiased to this basis is either the basis of $\sigma_2$ or
$\sigma_3$. The Holevo quantity $\chi\left(\tau\Big|\left\{\frac{1}{2}(\openone \pm \sigma_2)\right\}\right)$ pertaining to the basis of $\sigma_2$   is given by
\begin{equation}
\chi\left(\tau\Big|\left\{\frac{1}{2}(\openone \pm \sigma_2)\right\}\right)=1-h\left(\frac{1+|c_2|}{2}\right),
\end{equation}
 on the other hand, the Holevo quantity $\chi\left(\tau\Big|\left\{\frac{1}{2}(\openone \pm \sigma_3)\right\}\right)$ pertaining to the basis of $\sigma_3$   is given by
\begin{equation}
\chi\left(\tau\Big|\left\{\frac{1}{2}(\openone \pm \sigma_3)\right\}\right)=1-h\left(\frac{1+|c_3|}{2}\right).
\end{equation}
Since $|c_2| \ge |c_3|$, it follows that $Q^{\rightarrow}_2(\tau)$ is achieved by the basis of $\sigma_2$. On the other hand, $Q^{\rightarrow}_3(\tau)$ is achieved by the basis of $\sigma_3$. Therefore, $Q^{\rightarrow}_2(\tau)$ and $Q^{\rightarrow}_3(\tau)$ are given by 
\begin{equation}\label{Q2X}
Q^{\rightarrow}_2(\tau)=1-h\left(\frac{1+|c_2|}{2}\right)
\end{equation}
and
\begin{equation}\label{Q3X}
Q^{\rightarrow}_3(\tau)=1-h\left(\frac{1+|c_3|}{2}\right),
\end{equation}
respectively. $Q^{\rightarrow}_2(\tau)>0$ for any Bell-diagonal state having a nonzero discord, whereas
$Q^{\rightarrow}_3(\tau)>0$ if and only if the given Bell-diagonal state has a nonzero discord, with  $|c_3|>0$.

\subsection{Quantum steering} 
Let us consider a steering scenario where two spatially separated parties, Alice and Bob, share an unknown quantum system $\rho_{AB}$.  Alice 
performs a set of black-box measurements, and the Hilbert-space dimension of Bob's 
subsystem is known. Such a scenario is called one-sided device-independent since
Alice's measurement operators $\{M_{a|x}\}_{a,x}$, which are positive operator-valued measures (POVM), are unknown. The steering scenario 
is completely characterized by the set of unnormalized conditional states on Bob's
side $\{\sigma_{a|x}\}_{a,x}$, which is called an unnormalized assemblage. Each element
in the unnormalized assemblage is given by $\sigma_{a|x}=p(a|x)\rho_{a|x}$,  where $p(a|x)$ is the conditional probability of getting the outcome $a$ when Alice performs the measurement $x$;
$\rho_{a|x}$ is the normalized conditional state on Bob's side.
Quantum theory predicts that all valid assemblages should satisfy the following criteria:
\begin{equation}
\sigma_{a|x}={\rm Tr}_A{[ (M_{a|x} \otimes \openone) \rho_{AB}]} \hspace{0.5cm} \forall \sigma_{a|x} \in \{\sigma_{a|x}\}_{a,x}.
\end{equation}

\begin{definition}
In the above scenario,  Alice demonstrates  steerability to Bob in a 1SDI way \cite{WJD07}  if the assemblage does not have a local hidden state (LHS) model, i.e., if, for all $a$, $x$, there is no decomposition of $\sigma_{a|x}$ in the form,
\begin{equation} \label{WJDLHS}
\sigma_{a|x}=\sum_\lambda p(\lambda) p(a|x,\lambda) \rho_\lambda,
\end{equation}
where $\lambda$ denotes the classical random variable (also called local hidden variable (LHV)), which occurs with probability 
$p(\lambda)$; $\rho_{\lambda}$ are called local hidden states which satisfy $\rho_\lambda\ge0$ and
$\tr{\rho_\lambda}=1$.
\end{definition}

Next, we are going to present the formal definition of a different form of steering, called 1SSDI steering,
as introduced in Ref. \cite{JDS_PRA23}. In 1SSDI steering, the dimension of the resource on the untrusted side is constrained. But no other constraint is imposed on the untrusted side. In other words, the untrusted side is no longer fully device-independent here; it is semi device-independent.
\begin{definition}
In the steering scenario as described before, Alice demonstrates    steerability to Bob in a 1SSDI way if the assemblage does not have an LHS model with a restricted hidden variable dimension, i.e. if for all $a$, $x$, there is no decomposition of $\sigma_{a|x}$ in the form,
\begin{equation}\label{LHSdl}
\sigma_{a|x}=\sum^{d_\lambda-1}_{\lambda=0} p(\lambda) p(a|x,\lambda) \rho_\lambda,
\end{equation}
with $d_\lambda \le d_A$.
\end{definition}

The detection of 1SSDI steering from a nonsignaling (NS) box can be introduced following Ref. \cite{JDS_PRA23}. Suppose Bob performs a set of projective measurements $\{\Pi_{b|y}\}_{b,y}$ on $\{\sigma_{a|x}\}_{a,x}$. Then, the scenario is characterized by the set of measurement correlations or an NS box shared
between Alice and Bob $P(ab|xy)$:=$\{p(ab|xy)\}_{a,x,b,y}$, where $p(ab|xy)$ = ${\rm Tr}  ( \Pi_{b|y} \sigma_{a|x} )$.    
\begin{definition}
A bipartite box $P(ab|xy)$:=$\{p(ab|xy)\}_{a,x,b,y}$ arising in a steering scenario detects  steerability in a 1SSDI way if and only if there is no decomposition of the box in the form given by \begin{equation} \label{LHV-LHS}
p(ab|xy)=\sum^{d_\lambda-1}_{\lambda=0} p(\lambda) p(a|x, \lambda) p(b|y, \rho_{\lambda}) \hspace{0.3cm} \forall a,x,b,y,
\end{equation}
with $d_\lambda\le d_A$.
\end{definition}
Due to the additional constraint on the dimension of the resource at the untrusted side, that is, $d_\lambda\leq d_A$, there exist unsteerable assemblages satisfying Eq. (\ref{WJDLHS}) in the 1SDI context while violating Eq. (\ref{LHSdl}) in the 1SSDI context.
In fact, such assemblages have superunsteerability as defined in \cite{DBD+18}.
 \begin{definition}
Superunsteerability of a given unsteerable box holds if and only if there is no LHV-LHS model of the box given by
\begin{equation} 
p(ab|xy)=\sum^{d_\lambda-1}_{\lambda=0} p(\lambda) p(a|x, \lambda) p(b|y, \rho_{\lambda}) \hspace{0.3cm} \forall a,x,b,y,  \nonumber
\end{equation}
with $d_\lambda\le d_A$. 
\end{definition}
In the context of the resource theory for 1SDI steering \cite{GA15,DDJ+18}, all unsteerable assemblages and boxes produced from them, thus including superunsteerability, are free. However, in the context of 1SSDI scenarios, superunsteerability is not free as it provides the detection of the 1SSDI steerable resource.  The key difference in witnessing 1SSDI steering is that the classical randomness $\lambda$ is bounded by $d_\lambda\le d_A$, whereas the classical randomness $\lambda$ is freely available to witness the 1SDI steerability. See Ref. \cite{JKC+24} for a discussion of the resource theory of 1SSDI steering.

Let us now present an example of superunsteerability \cite{DBD+18}  where Alice performs two dichotomic black-box measurements, and Bob performs two dichotomic trusted qubit measurements in MUBs.  
Consider the unsteerable white noise-BB84 family given by  
\begin{equation}
\label{BB84}
P(ab|xy) = \frac{1 + (-1)^{a \oplus b \oplus x.y} \delta_{x,y} V }{4},
\end{equation}
with $0 < V \leq 1/\sqrt{2}$. The above box can be produced by a two-qubit Werner state with an entanglement or a nonzero quantum discord for appropriate local non-commuting measurements. On the other hand, it cannot be simulated by an LHV-LHS model with $d_\lambda=2$ as shown in Ref. \cite{DBD+18}.

In Ref. \cite{JKC+24}, the relationship between bipartite coherence, quantum discord, and superunsteerability has been explored.
  In this relationship,
 it has been demonstrated that the global coherence, i.e., discord with $L_R > d_{min}$, is necessary to demonstrate superunsteerability.

In Ref. \cite{JDK+19}, a quantity called Schr\"{o}dinger strength (SS) was defined to quantify the nonclassicality of unsteerable boxes which demonstrate superunsteerability.
Note that a box $P(ab|xy)$ in the $n$-setting steering scenario is the set of joint probabilities $p(ab|xy)$ for all possible $a$, $b$ and for all $x$ $\in$ $\{0, 1, 2, ..., n-1\}$ and $y$ $\in$ $\{0, 1, 2, ..., n-1\}$. In this scenario, any box $P(ab|xy)$ can be decomposed into a convex mixture of a steerable part in 1SDI scenario and an unsteerable part in 1SDI scenario,
\be \label{steerstrength}
 P(ab|xy)=p P_{S}(ab|xy)+(1-p) P_{US}(ab|xy),    
\ee
where $P_{S}(ab|xy)$ is a steerable box in 1SDI scanrio  and $P_{US}(ab|xy)$  is an unsteerable box in 1SDI scenario, which may be superunsteerable; $0 \leq p \leq 1$. 
The weight of the box $P_{S}(ab|xy)$ maximized over all possible decompositions of 
the  form (\ref{steerstrength}) is called the Schr\"{o}dinger strength (SS) of the box $P(ab|xy)$, given by 
\be \label{SSdefi}
{\rm SS}_{n} \Big( P(ab|xy) \Big):=\max_{\text{decompositions}}p.
\ee
Here, $0 \leq {\rm SS}_{n} \leq 1$ (since, $0 \leq p \leq 1$). 
The optimal decomposition that gives the SS of the box is called the canonical decomposition in which the steerable part is an extremal steerable box in 1SDI scenario. An extremal steerable box in 1SDI scenario cannot be written as a convex mixture of boxes in the given scenario.


For $ 0 < V \le 1/\sqrt{2}$, the box (\ref{BB84}) can be decomposed as a convex mixture 
of the extremal (Ext) steerable box  and an unsteerable box in 1SDI scenario as follows \cite{JDK+19}:
\be \label{BB84dec1}
P(ab|xy)=V P_{\rm Ext}(ab|xy)+ (1-V) P_N(ab|xy),
\ee
where $P_N(ab|xy)$ is the white noise for which $P_N(ab|xy)=1/4$ for all $a,b,x,y$. In fact, $P_N(ab|xy)$ is an unsteerable box not only in 1SDI scenario, but also in 1SSDI scenario. On the other hand, 
\be
P_{\rm Ext}(ab|xy)=\frac{1+(-1)^{a+b+xy}\delta_{x,y}}{4}, \quad a,b,x,y \in \{0,1\}.
\ee
The above decomposition is optimal  implying
the SS as $V$ \cite{JDK+19}.

In a $n$-setting steering scenario where Alice performs $n$ black-box dichotomic measurements, and Bob performs $n$ characterized dichotomic measurements, let Alice's measurement settings and Bob's measurement settings are denoted by $A_x$ and $B_y$ respectively, where $A_x$ $\in$ $\tilde{A}$ = $\{A_0, A_1, A_2, ..., A_{n-1} \}$ and $B_y$ $\in$ $\tilde{B}$ = $\{ B_0, B_1, B_2, ..., B_{n-1} \}$. The sets of Alice's and Bob's measurement settings are denoted by $\tilde{A}$ and $\tilde{B}$, respectively. 
 We will now present the definition of SS of a bipartite state in an $n$-setting scenario.
\begin{definition}
 The SS of a bipartite state $\rho$ in $n$-setting scenario is defined as:
 \be
 {\rm SS}_{n}(\rho)=\max_{\tilde{A}, \tilde{B}} {\rm SS}_{n}\Big(P(ab|xy;\rho) \Big).
 \ee
Here ${\rm SS}_{n}\Big(P(ab|xy;\rho) \Big)$ is the SS of the box $P(ab|xy;\rho)$ = $\{ p(ab|xy;\rho) \}_{a,x,b,y}$ as defined in Eq. (\ref{SSdefi}), where $p(ab|xy;\rho)$ = Tr$[ (M_{a|x} \otimes M_{b|y}) \rho ]$ is the joint probability distribution of getting the outcomes $a$ and $b$ when measurements $A_x$ and $B_y$ are performed locally by Alice and Bob, respectively, on the shared bipartite state $\rho$. $M_{a|x}$ is the measurement operator corresponding to the measurement settings $A_x$ and outcome $a$. $M_{b|y}$ is defined similarly. Here the maximization is taken over all possible sets $\tilde{A}$ = $\{A_0, A_1, A_2, ..., A_{n-1} \}$ and $\tilde{B}$ = $\{ B_0, B_1, B_2, ..., B_{n-1} \}$.
\end{definition}

In Ref. \cite{JDK+19}, the SS of two-qubit Bell-diagonal states (\ref{BD}) have been evaluated in the following two-setting and three-setting scenarios. 
(i) Alice performs two black-box dichotomic measurements, and Bob performs 
projective qubit measurements corresponding to two  MUBs
and 
(ii) Alice performs three black-box dichotomic measurements, and Bob performs 
projective qubit measurements corresponding to three MUBs.

The SS of two-qubit Bell-diagonal states in the two-setting scenario has been evaluated as follows \cite{JDK+19}:
\begin{prop}
The SS of the Bell-diagonal states defined in Eq.(\ref{BD}) in the two-setting steering scenario is given by
\ba
{\rm SS}_{2}(\tau)=|c_2|. \label{SSBD2}
\ea 
\end{prop}
On the other hand, in the three-setting scenario,  the Schr\"{o}dinger strength of two-qubit Bell-diagonal states  has been evaluated as follows \cite{JDK+19}:  
\begin{prop}
The SS of the Bell-diagonal states defined in Eq.(\ref{BD}) in the three-setting steering scenario  is given by
\ba
{\rm SS}_{3}(\tau)=|c_3|. \label{SSBD3}
\ea 
\end{prop}


\subsection{Quantum steering ellipsoids for  two-qubit states}\label{II}
	Consider a two-qubit state
	\begin{align}
		\rho_{AB} =  \frac{1}{4} \bigg( & \mathbbm{1}_A \otimes \mathbbm{1}_B+{\mf{a}} \cdot \boldsymbol{\sigma} \otimes \mathbbm{1}_B+\mathbbm{1}_A \otimes {\mf{b}} \cdot \boldsymbol{\sigma}     \nn \\  
						        &  +\sum^3_{j,k=1}T_{jk}\sigma_j\otimes\sigma_k \bigg), 
		\label{state}
	\end{align}
where $\boldsymbol{\sigma}\equiv (\sigma_1, \sigma_2, \sigma_3)$ denotes the vector of Pauli operators and $\mathbbm{1}_A, \mathbbm{1}_B$ are identity operators. Here ${\mf{a}} = (a_1, a_2, a_3)$ and ${\mf{b}} = (b_1, b_2, b_3)$ are the Bloch vectors of Alice's and Bob's reduced states and $T_{jk}$ are elements of the spin correlation matrix,  i.e.,
	\begin{align}
	 a_j&:=\tr{\rho_{AB}(\sigma_j\otimes \mathbbm{1}_B)}, ~~~ b_k:=\tr{\rho_{AB}(\mathbbm{1}_A\otimes \sigma_k)}, \nn \\
	 T_{jk}&:=\tr{\rho_{AB}(\sigma_j\otimes \sigma_k)}, \qquad (j, k=1, 2, 3). \label{notation}
	\end{align}

	Let Alice performs a general POVM which has an element $M_a=a_0\left(\mathbbm{1}_A+\mf{m}_a\cdot\boldsymbol{\sigma} \right)$, with $a_0\geq 0$ and $|\mf{m}_a|\leq 1$ corresponding to an outcome $a$. This outcome is obtained with 
probability  
	\beq
	p_{M_a}=\tr{\rho_{AB}\, (M_a\otimes \mathbbm{1}_B)}=a_0\left(1+{\mf{a}}\cdot \mf{m}_a\right), \nn
	\eeq
	leading to the steered state 
	\begin{align}\label{reducedrho}
	\rho^{M_a}_B&=\frac{{\rm Tr}_A[\rho_{AB}\,(M_a\otimes \mathbbm{1}_B)]}{p_{M_a}} \nn \\
    &=\frac{1}{2}\left(\mathbbm{1}_B+\frac{({\mf{b}}+T^\top\mf{m}_a)\cdot \boldsymbol{\sigma}}{1+{\mf{a}}\cdot\mf{m}_a}\right) 
	\end{align}
	for Bob's qubit. Here, $T^\top$ is the transpose of the matrix $T=(T_{jk})$.

    \begin{definition}
        Considering all possible local measurements by Alice, it follows that the corresponding set of Bob's steered states is represented by the set of Bloch vectors
	\beq \label{reduced}
	 \mathcal{E}_{B|A}:=\left\{ \frac{{\mf{b}}+T^\top\mf{m}_a}{1+{\mf{a}}\cdot\mf{m}_a}:|\mf{m}_a|\leq 1 \right\}.
	 \eeq
	Quantum steering ellipsoid (QSE) is defined as the above set which  forms a (possibly degenerate) ellipsoid ~\cite{JPJR14}. The subscript $B|A$ indicates the steering ellipsoid for Bob that is generated by Alice's local measurements.
    \end{definition}
	  Similarly, a QSE for Alice is generated by Bob's local measurements, denoted by $\mathcal{E}_{A|B}$, can be defined.
	The  QSE $\mathcal{E}_{B|A}$ is fully determined by its centre,
	\beq
	{\mf{c}}_{B|A}=\frac{{\mf{b}}-T^\top{\mf{a}}}{1-a^2}, \nn
	\eeq
	and its orientation matrix,
	\beq
	Q_{B|A}=\frac{1}{1-a^2}\left(T-{\mf{a}}{\mf{b}}^\top\right)^\top\left(I+\frac{{\mf{a}}{\mf{a}}^\top}{1-a^2}\right)\left(T-{\mf{a}}{\mf{b}}^\top\right), \nn
	\eeq
	where the eigenvalues and corresponding eigenvectors of $Q_{B|A}$ describe the squared lengths of the ellipsoid's semiaxes and their orientations~\cite{JPJR14}. 
	
	The QSE $\mathcal{E}_{B|A}$, together with the reduced Bloch vectors ${\mf{a}}$ and ${\mf{b}}$,  provides a faithful representation of any two-qubit state up to local unitary operations on Alice's qubit ~\cite{JPJR14}. Additionally, its various geometric properties encode useful information about the state.
  	 For example, the state is separable if and only if its QSE can be nested in a tetrahedron that is, in turn, nested in the Bloch sphere~\cite{JPJR14}. 
    
    Furthermore, two-qubit separable states can be classified into states with a nonzero volume for their QSEs and states with degenerate QSEs. A nonzero volume for QSEs of the two-qubit states is stronger than a nonvanishing discord. A degenerate QSE can be $2$-dimensional (a steering pancake) or $1$-dimensional (a steering needle) or trivially $0$-dimensional. Within QSE, a two-qubit state has zero discord from Alice to Bob if and only if her QSE degenerate to a segment of a diameter.
    
   In cases of steering pancakes and steering needles, steering can be either complete or incomplete. 
   \begin{definition}
         Complete steering of Bob occurs if and only if for
    any convex decomposition of the reduced state of Bob in terms of the states in $\mathcal{E}_{B|A}$, there exists a POVM on Alice's side that steers Bob to this ensemble. 
   \end{definition}
   For any separable state whose QSE is either a pancake or needle, the complete steering of Bob by Alice occurs if the reduced state of Alice is maximally mixed \cite{JPJR14}.

\section{Results}

\subsection{Operational SCMUB via $1$SSDI steering}
Here,  we identify which discordant states have the simultaneous correlations in two MUBs, i.e., $Q^{\rightarrow}_2>0$. We then show that any $Q^{\rightarrow}_2>0$ of the given bipartite state can be operationally identified as exhibiting $1$SSDI steering in a two-setting scenario.

\begin{lem}
    A bipartite state in $\mathbb{C}^{d_A}\otimes\mathbb{C}^{d_B}$ has $Q^{\rightarrow}_2=0$ if and only if $L_R$ in Eq. (\ref{LR}) satisfies $L_R \le d_{\min}$.
\end{lem}
\begin{proof}
If a discordant state has $L_R \le d_{\min}$, then it has correlation  in a single basis. On the other hand, if a discordant state has $Q^{\rightarrow}_2=0$, it has
$L_R \le d_{\min}$ since it has correlation in only a single basis.
\end{proof}

From the above lemma, we obtain the following result.
\begin{thm}
A bipartite state in $\mathbb{C}^{d_A}\otimes\mathbb{C}^{d_B}$ has $Q^{\rightarrow}_2>0$  if and only if it has global coherence, i.e., $L_R > d_{\min}$.
\end{thm}
\begin{proof}
Since $L_R > d_{\min}$ implies that it has correlation in two bases, it follows that any discordant state with global coherence has $Q^{\rightarrow}_2>0$.
\end{proof}

\begin{cor}\label{corQ21SSDI}
Any nonzero $Q^{\rightarrow}_2$ of a bipartite state in $\mathbb{C}^{d_A}\otimes\mathbb{C}^{d_B}$ can be used to demonstrate 1SSDI steering from Alice to Bob in a two-setting steering scenario.
\end{cor}
\begin{proof}
 To prove this, we use the following facts.
 If an assemblage $\{\sigma_{a|x}\}$ has an LHS model as in Eq. (\ref{LHSdl}),
 then it can be reproduced using a CQ state in $\mathbb{C}^{d_A}\otimes\mathbb{C}^{d_B}$ \cite{JD23,JKC+24}. Note that any CQ state has correlation in only a single basis. Now, consider the assemblage $\{\sigma_{a|x}\}$ arising from any bipartite state in $\mathbb{C}^{d_A}\otimes\mathbb{C}^{d_B}$ with a nonzero $Q^{\rightarrow}_2$ for two MUBs that give rise to $Q^{\rightarrow}_2>0$.
Such an assemblage cannot be reproduced using a CQ state   in $\mathbb{C}^{d_A}\otimes\mathbb{C}^{d_B}$ because the CQ state has correlation in only a single basis. Therefore, any nonzero $Q^{\rightarrow}_2$ can be used to demonstrate $1$SSDI  steering. 
\end{proof}
Thus, any nonzero $Q^{\rightarrow}_2$ has an operational characterization in a two-setting steering scenario to demonstrate $1$SSDI steering.

\subsection{Operational SCMUB of two-qubit states}

\subsubsection{Bell-diagonal states}
We study the quantification of 1SSDI steerability in Bell-diagonal two-qubit states and their relationships with the quantification of SCMUB of these states. For this purpose, we use the quantification of superunsteerability using ${\rm SS}_n$ mentioned earlier to quantify 1SSDI steerability in the two-setting and three-setting scenarios and show that they provide an operational characterization of $Q^{\rightarrow}_2$ and $Q^{\rightarrow}_3$, respectively.

For the Bell-diagonal states, the values of quantifiers in $\{C_1^{\rightarrow},Q^{\rightarrow}_2, Q^{\rightarrow}_3\}$ are given by Eqs. (\ref{C1BD}), (\ref{Q2X}) and (\ref{Q3X}), respectively.
We now proceed to derive analytical relationships between the measures of SCMUB and Schr\"{o}dinger strengths of $\tau$ as given by Eqs. (\ref{SSBD2}) and (\ref{SSBD3}).

Using Eqs. (\ref{SSBD2}) and (\ref{Q2X}), it is seen that ${\rm SS}_{2}(\tau)$ and $Q^{\rightarrow}_2(\tau)$ are related to each other as follows:
\be
\label{Q2SS2}
Q^{\rightarrow}_2(\tau)=1-h\left(\frac{1+{\rm SS}_{2}(\tau)}{2}\right).
\ee
On the other hand,
 the following relationship can be easily obtained using Eqs. (\ref{SSBD3}) and (\ref{Q3X}):
\be \label{SSQ3X}
Q^{\rightarrow}_3(\tau)=1-h\left(\frac{1+{\rm SS}_{3}(\tau)}{2}\right).
\ee

From the above two analytical relationships, we have obtained the following result.
\begin{result}
  $Q^\rightarrow_2(\tau)$ and  $Q^\rightarrow_3(\tau)$ are a monotonically increasing function of ${\rm SS}_{2}(\tau)$ and ${\rm SS}_{3}(\tau)$, respectively.
\end{result}
\begin{proof}
Note that if a function $f(x)$ is continuous and differentiable in any given interval $a \leq x \leq b$, then $f(x)$ is a monotonically increasing function of $x$ provided the derivative of $f(x)$ with respect to $x$ is not negative, i.e., $\frac{df(x)}{dx} \geq 0$ for all $x \in [a,b]$. 
 By differentiating the left hand side of Eq. (\ref{Q2SS2}) with respect to ${\rm SS}_{2}(\tau)$, it  can be checked that $\frac{d Q^\rightarrow_{2}}{d {\rm SS}_{2}(\tau)} \geq 0$ for all values of ${\rm SS}_{2}(\tau)$ which implies that $Q^\rightarrow_{2}$ is a monotonically increasing function of ${\rm SS}_{2}(\tau)$.
  Similar to the above two-setting case, it can be checked from Eq. (\ref{SSQ3X}) that the first derivative of $Q^{\rightarrow}_3$ with respect to ${\rm SS}_{3}(\tau)$ is non-negative, i.e., $\frac{d\mathcal{Q}_{3}}{d{\rm SS}_{3}(\tau)}\geq 0$ for all values of ${\rm SS}_{3}(\tau)$;  i.e., in the three-setting case too, $Q^{\rightarrow}_3$ is a monotonically increasing function of ${\rm SS}_{3}(\tau)$.
  \end{proof}
Thus, we have shown that the   Schr\"{o}dinger strengths ${\rm SS}_{2}(\tau)$ and ${\rm SS}_{3}(\tau)$
provide operational characterization of  $Q^{\rightarrow}_2(\tau)$ and $Q^{\rightarrow}_3(\tau)$, respectively.  
    
  We visualize the above result using QSE formalism to understand it geometrically.   For the Bell-diagonal states (\ref{BD}), the QSE of Bob  is given by
\begin{align}
\frac{x^2}{\ell_1^2}+\frac{y^2}{\ell_2^2}+\frac{z^2}{\ell_3^2}=1,
\end{align}
where $\ell_1=c_1$, $\ell_2=c_2$,
and $\ell_3=c_3$. The center of the ellipsoid is at the origin.
The $x$, $y$ and the $z$ semi-axes lengths of the ellipsoid is given by
\begin{align}
|l_1|&=|c_1|, \quad |l_2|=|c_2|, \quad |l_3|=|c_3|.
\end{align}
The reduced state $\rho^B$ corresponds to the point $B=(0,0,0)$, which is the origin. The above steering ellipsoid's normalized volume as defined in Ref. \cite{CMHW16} is given by $v_{B|A}=|c_1c_2c_3|$.
Note that if $c_1 -c_2=c_1+c_2$ or $(1+c_3)^2=(1-c_3)^2$, the ellipsoid degenerates to an ellipse or a line segment. 

\begin{figure}[t!]
\begin{center}
\includegraphics[width=8cm]{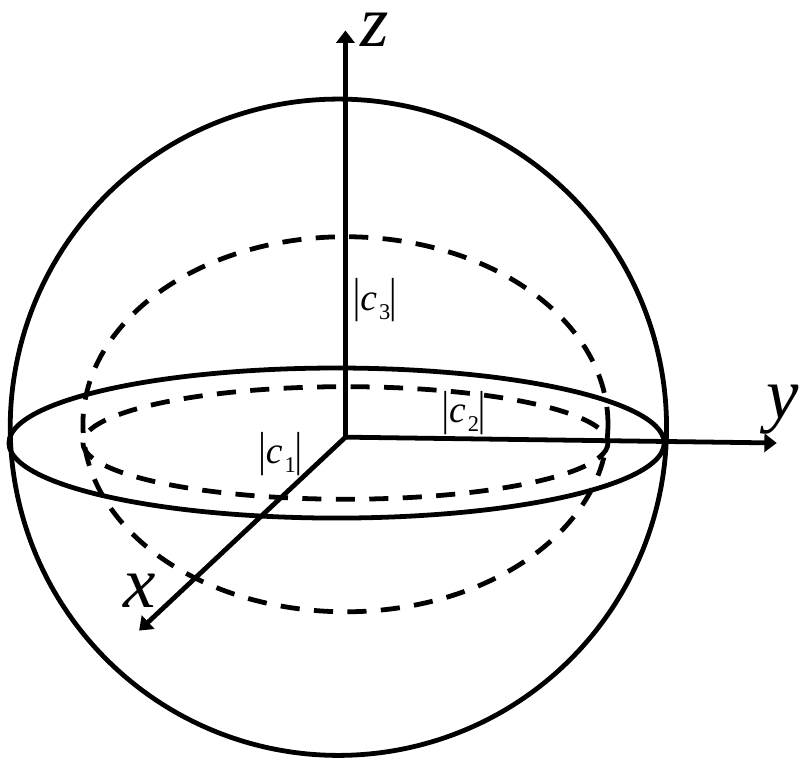}
\end{center}
\caption{The dashed ellipsoid inside the Bloch sphere represents QSE of the Bell-diagonal states with the semi-axes lengths satisfy $|c_1| \ge |c_2| \ge |c_3|$. Quantifiers $C_1^{\rightarrow}(\tau)$,  $Q^\rightarrow_2(\tau)$ (${\rm SS}_{2}(\tau)$), and $Q^\rightarrow_3(\tau)$ (${\rm SS}_{3}(\tau)$)  are a monotonically increasing function of the lengths of these three semi-axes, respectively.
\label{Fig:QSE_BD}}
\end{figure}
 Since $\frac{dC_1^{\rightarrow}(\tau)}{d|l_1|} \ge 0$, $\frac{dQ^\rightarrow_2(\tau)}{d|l_2|} \ge 0$
and $\frac{dQ^{\rightarrow}_3(\tau)}{d|l_3|} \ge 0$ for all values of $|l_i|$,
with $i=1,2,3$, respectively,  we have obtained the following observation.
\begin{observation} \label{obsmono}
    $C_1^{\rightarrow}(\tau)$, $Q^\rightarrow_2(\tau)$ (or ${\rm SS}_{2}(\tau)$) and  $Q^\rightarrow_3(\tau)$ (or ${\rm SS}_{3}(\tau)$) are a monotonically increasing function of the lengths of the semi-axes of the QSE in descending order, respectively.
\end{observation}
Also, note that $Q^\rightarrow_3(\tau)>0$ or ${\rm SS}_{3}(\tau)>0$ if and only if the QSE of Bob has a nonvanishing volume. See Fig. \ref{NewFig} for the illustration of the above observation. 

To study the monotonic relationships  mentioned in the Observation \ref{obsmono}, let us consider the following two-parameter family of Bell-diagonal states:
\begin{align}\label{BD2para}
\tau(p,q)&=p \ket{\beta_{00}}\bra{\beta_{00}} +\frac{q}{2}( \ket{\beta_{10}}\bra{\beta_{10}}+\ket{\beta_{11}}\bra{\beta_{11}}) \nonumber \\
&+\frac{(1-p-q)}{4}\sum_{ab}\ket{\beta_{ab}}\bra{\beta_{ab}},
\end{align}
with $0\leq p, q \leq 1$, $p \ge q$ $p+q\leq 1$. The state $\tau(p,q)$ is separable if $(3p-q) \le 1$, otherwise, it is entangled \cite{Per96, HHH96} and Wootters’s concurrence \cite{Woo98} of $\tau(p,q)$ is given by $\min\{0,(3p-q-1)/2\}$. For this subclass of Bell-diagonal states, the classical correlation and the discord are given by
\begin{align}
\label{BD2cd}
\begin{split}
C^{\rightarrow}_1(\tau(p,q))&=1-h\left(\frac{1+p}{2}\right), \\
D^{\rightarrow}(\tau(p,q))&=\sum_{ab} \lambda_{ab} \log_2(4\lambda_{ab})  -C^{\rightarrow}_1(\tau(p,q)),
\end{split}
\end{align}
where $\lambda_{00}=(1+3p-q)/4$, $\lambda_{01}=\lambda_{10}=(1-p+q)/4$ and $\lambda_{11}=(1-p-q)/4$, and the SCMUB measures are given by
\begin{align}
\label{BD2cqs}
\begin{split}
Q^{\rightarrow}_2(\tau(p,q))&=1-h\left(\frac{1+{\rm SS}_{2}(\tau(p,q))}{2}\right), \\
Q^{\rightarrow}_3(\tau(p,q))&=1-h\left(\frac{1+{\rm SS}_{3}(\tau(p,q))}{2}\right), 
\end{split}
\end{align}
where ${\rm SS}_{2}(\tau(p,q))=p$ and ${\rm SS}_{3}(\tau(p,q))=p-q$.  
See Figs. \ref{Fig:QSE_BD2} and \ref{NewFig} for the illustration of the Observation \ref{obsmono} in the context of the family of states given by Eq. (\ref{BD2para}). 

For $p>q$ and $(3p-q) \le 1$, $\tau(p,q)$ is separable, but  has a nonzero QSE volume and has a nonzero $Q^{\rightarrow}_2(\tau(p,q))$ as well as a nonzero $Q^{\rightarrow}_3(\tau(p,q))$, which are a monotonically increasing function of ${\rm SS}_{2}(\tau(p,q))$ and ${\rm SS}_{3}(\tau(p,q))$, respectively, as in Eq. (\ref{BD2cqs}).  On the other hand, for $p=q$, $\tau(p,q)$ has vanishing QSE volume, but the discord $D^{\rightarrow}(\tau(p,q))$ is nonzero if $p \ne 0$  and has SCMUB  only by a nonzero $Q^{\rightarrow}_2(\tau(p,q))$.

\begin{figure}[t!]
\begin{center}
\includegraphics[width=8cm]{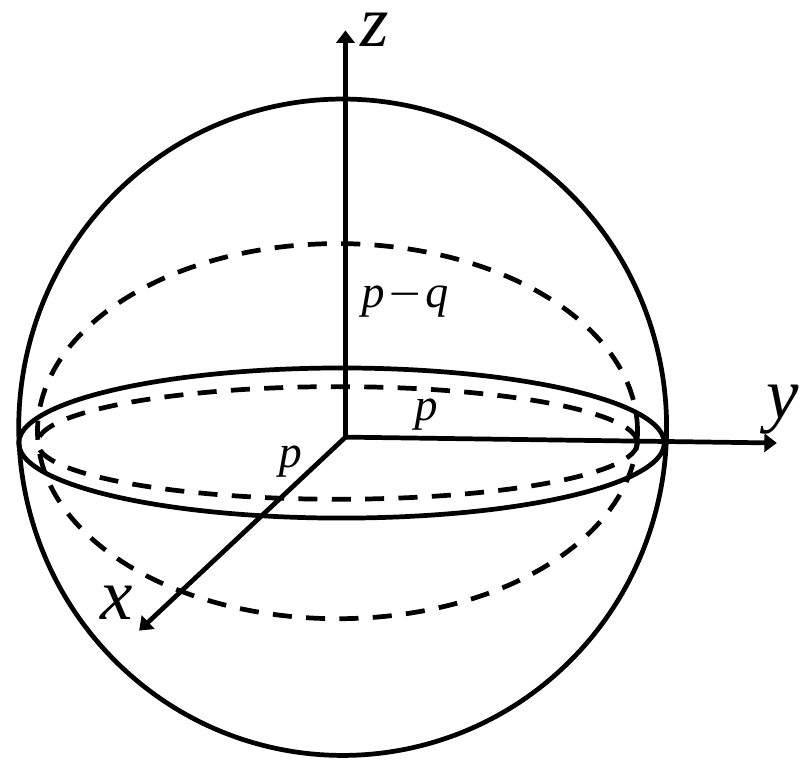}
\end{center}
\caption{The dashed ellipsoid inside the Bloch sphere represents QSE of the two-parameter Bell-diagonal states given by Eq. (\ref{BD2para}). 
\label{Fig:QSE_BD2}}
\end{figure}

\begin{figure}[t!]
\begin{center}
\includegraphics[width=8.5cm]{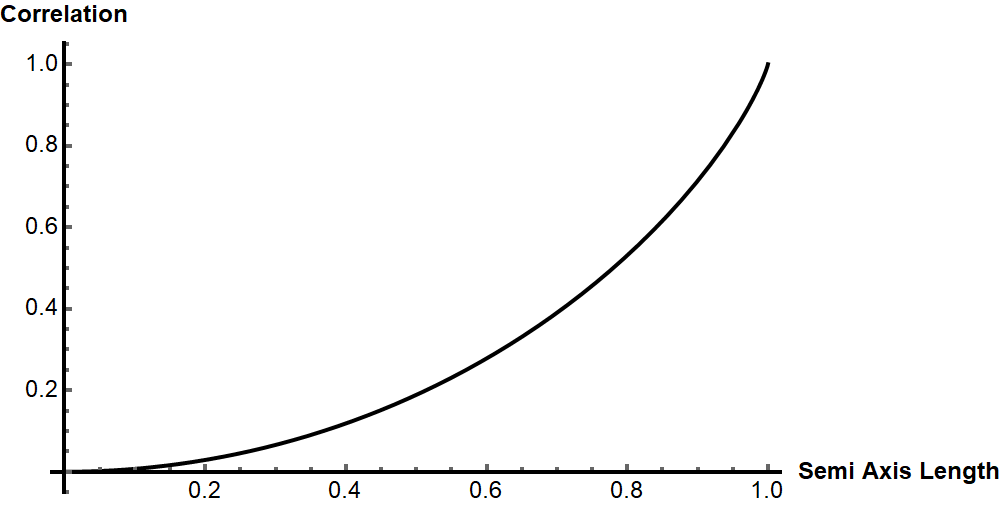} 
\end{center}
\caption{The variation of three different types of correlations $C_1^{\rightarrow}(\tau(p,q)$,  $Q^\rightarrow_2(\tau(p,q))$, and $Q^\rightarrow_3(\tau(p,q))$) with respect to lengths of the three semi-axes ($p$, $p$ and $p-q$ respectively) of the QSE. The same figure depicts the plot of (1) $C_1^{\rightarrow}(\tau(p,q))$ versus $p$, (2) $Q^\rightarrow_2(\tau(p,q))$ versus $p$, and (3) $Q^\rightarrow_3(\tau(p,q))$ versus $p-q$. This shows that the correlations $C_1^{\rightarrow}(\tau(p,q))$,  $Q^\rightarrow_2(\tau(p,q))$, and $Q^\rightarrow_3(\tau(p,q))$   are monotonically increasing functions of the lengths of these three semi-axes respectively.}
\label{NewFig}
\end{figure}

\subsubsection{Classification of  two-qubit separable states}
We  provide a different classification of two-qubit discordant states based on SCMUB quantified by $Q^{\rightarrow}_2$ and $Q^{\rightarrow}_3$. We also study its relationship with ${\rm SS}_2$ and ${\rm SS}_3$, respectively. To this end, we use the QSE formalism to give a geometric understanding of the classification.

Any separable two-qubit state can be written as a convex mixture of at most four product states as follows \cite{JPJR14}:
\be \label{sep2q}
\rho_{AB}=\sum^n_{i=1} p_i \alpha_i \otimes \beta_i,
\ee
with $n \le 4$. 
Any separable state with $n \le 2$ in Eq. (\ref{sep2q}) cannot be used to demonstrate 1SSDI steerability. This follows that any such state always has an LHV-LHS model as in Eq. (\ref{LHV-LHS}) with $d_\lambda \le 2$. Also, such states do not have global coherence, which is required to demonstrate 1SSDI steerability \cite{JKC+24}. 
We now have obtained the following result.
\begin{result} \label{1dimq20}
    For any two-qubit discordant state whose QSE is $1$-dimensional, i.e., a separable two-qubit state of the form given by Eq. (\ref{sep2q}) with $n=2$, the SCMUB  captured by $Q^\rightarrow_2$ vanishes.
\end{result}
\begin{proof}
  This follows because any two-qubit discordant state whose QSE is $1$-dimensional does not have global coherence, implying that it has correlation in only a single basis. Whereas, for a nonzero $Q^\rightarrow_2$,
  the discordant state should have correlation in more than one basis.
  \end{proof}
  To illustrate the above result, consider the following  rank-$2$ separable state:
  \begin{equation}
\label{stateowsu}
\rho^{\texttt{rank}-2}_{1\texttt{way}D} = \frac{1}{2} \Big( |00\rangle \langle 00| + |+1 \rangle \langle +1| \Big) ,
\end{equation}
where, $|0\rangle$ and $|1\rangle$ are the eigenstates of the operator $\sigma_3$ corresponding to the eigenvalue $+1$ and $-1$ respectively; $|+\rangle$ is the eigenstate of the operator $\sigma_1$ corresponding to the eigenvalue $+1$. The above state has the discord as given by  $D^{\rightarrow}(\rho^{\texttt{rank}-2}_{1\texttt{way}D}) \approx  0.2018$ \cite{BPP15}
and $D^{\leftarrow}(\rho_{AB})=0$  since it is not a CQ state but a quantum-classical state \cite{HV01,OZ01}. Hence, it is one-way discordant ($1\texttt{way}D$).
On the other hand, if we consider the following another rank-$2$ separable state:
\begin{equation}
\label{statesu3}
\rho^{\texttt{rank}-2}_{2\texttt{way}D} = \frac{1}{2} \Big( |00\rangle \langle 00| + |++ \rangle \langle ++| \Big),
\end{equation}
the state has a nonzero quantum discord value from both Alice to Bob and Bob to Alice as $D^{\rightarrow}(\rho^{\texttt{rank}-2}_{2\texttt{way}D})=D^{\leftarrow}(\rho^{\texttt{rank}-2}_{2\texttt{way}D}) \approx  0.1442$ \cite{SYF+11}. Hence, it is two-way discordant ($1\texttt{way}D$). In the context of Result \ref{1dimq20}, we now show the following with explicit calculations.
\begin{observation}
 For the states $\rho^{\texttt{rank}-2}_{1\texttt{way}D}$  and $\rho^{\texttt{rank}-2}_{2\texttt{way}D}$, the SCMUB measure $Q^\rightarrow_2=0$.
\end{observation}
\begin{proof}
As shown in Ref. \cite{SYF+11}, the optimal bases that give rise to $C^\rightarrow_1$ of such rank-$2$ states are the bases of the observables $\frac{\mp\sigma_z\pm \sigma_x}{\sqrt{2}}$. For such basis, the classical correlations of the states are obtained as $C^\rightarrow_1(\rho^{\texttt{rank}-2}_{1\texttt{way}D})=1-h\left(\frac{1+1/\sqrt{2}}{2}\right)$ and $C^\rightarrow_1(\rho^{\texttt{rank}-2}_{2\texttt{way}D})=h\left(\frac{2+\sqrt{2}}{4}\right)-h\left(\frac{1+\sqrt{3}/2}{2}\right)$, respectively. On the other hand, the Holevo quantity of the ensemble obtained from these states for a basis that is mutually unbiased to the above-mentioned $C^\rightarrow_1$-bases of the state vanishes. For instance, for the basis associated with the observable $\frac{\sigma_z + \sigma_x}{\sqrt{2}}$, it has been checked that these states have
$\chi\Big(\rho^{\texttt{rank}-2}_{1\texttt{way}D}\Big|\Big\{\frac{1}{2}\Big(\one \pm \frac{\sigma_z + \sigma_x}{\sqrt{2}}\Big)\Big\}\Big)=\chi\Big(\rho^{\texttt{rank}-2}_{2\texttt{way}D}\Big|\Big\{\frac{1}{2}\Big(\one \pm \frac{\sigma_z + \sigma_x}{\sqrt{2}}\Big)\Big\}\Big)=0$. 
This implies that the simultaneous correlations in two MUBs of the states $\rho^{\texttt{rank}-2}_{1\texttt{way}D}$  and $\rho^{\texttt{rank}-2}_{2\texttt{way}D}$ vanishes, i.e., $Q^\rightarrow_2(\rho^{\texttt{rank}-2}_{1\texttt{way}D})=Q^\rightarrow_2(\rho^{\texttt{rank}-2}_{2\texttt{way}D})=0$.
\end{proof}
The above observation illustrates that a nonzero $Q^{\rightarrow}_2$ is stronger than a nonzero  $D^{\rightarrow}$ since certain discordant states have vanishing $Q^{\rightarrow}_2$.

Next, we consider the separable states of the form (\ref{sep2q}) with $n=4$. Within this class of states, we consider the following subset of states:
\be \label{sepn=4}
\rho_{AB}=\sum^4_{i=1} p_i \rho_i \otimes \ketbra{\psi_i}{\psi_i},
\ee
with $\rho_B=\frac{\openone}{2}$. The QSE of these states is $3$-dimensional \cite{JPJR14}. We have now obtained the following result.
\begin{result}
Any discordant state whose QSE is $3$-dimensional, which has the form given by Eq. (\ref{sepn=4}),  has a nonvanishing $Q^\rightarrow_2$  as well as a nonvanishing  $Q^\rightarrow_3$.  
\end{result}
\begin{proof}
In Ref. \cite{JD23}, it has been shown that if one of the reduced states of any discordant state is maximally mixed, it has superunsteerability in the two-setting scenario. From this, it follows that any discordant state that has the form given by Eq. (\ref{sepn=4}) exhibits 1SSDI steerability implying that it has a nonzero $Q^\rightarrow_2$. Since the QSE of any such state is $3$-dimensional, it also follows that it has simultaneous correlations in three MUBs as quantified by $Q^\rightarrow_3$.  
\end{proof}
An example of a state that satisfies the property of the above result is the Bell-diagonal state with $c_1=c_2=c_3=1/3$,   Before we express this state  $\tau'$ 
in the form given by Eq. (\ref{sepn=4}), we introduce the state $\ket{\theta,\phi}$ and the set $\{\mathcal{Z}\}$ as follows.
Denoting by
    \begin{equation} \label{eq:desc_sep_1}
    \ket{\theta,\phi} := \cos\left(\frac{\theta}{2}\right) \ket{0}
                        + \exp(i \phi)\sin\left(\frac{\theta}{2}\right) \ket{1} \,,
    \end{equation}
an arbitrary pure state in $\mathbb{C}^2$.
Defining ${\mathcal{Z}=\{\ket{0,0},\ket{\theta^*,0},
\ket{\theta^*,\frac{2\pi}{3}},\ket{\theta^*,\frac{4\pi}{3}}\}}$,
with ${\theta^*=\arccos(-\frac{1}{3})}$, the state
$\tau'$, which is the above-mentioned Bell-diagonal state, can be decomposed as follows \cite{BPP15}:
    \begin{equation} \label{eq:desc_sep_2}
    \tau' = \frac{1}{4} \sum_{k=1}^4{Z_k\otimes Z_k} \,,
    \end{equation}
with ${Z_k=\ket{z_k}\bra{z_k}}$ and ${\ket{z_k}\in\mathcal{Z}}$, which has a nonzero discord value as $D^{\rightarrow}( \tau') = 0.3333$  \cite{BPP15}.  For this class of state, Eqs. (\ref{Q2X}) and (\ref{Q3X}) with $c_2=c_3=1/3$ imply operational quantification of simultaneous correlations in two and three MUB via the Schr\"{o}dinger strengths ${\rm SS}_{2}(\tau')=1/3$ and ${\rm SS}_{3}(\tau')=1/3$, respectively.

Finally, we consider the separable states (\ref{sep2q}) with $n=3$. Within this class of states, we consider the following subset of states:
\be \label{sepn=3}
\rho_{AB}=\sum^3_{i=1} p_i \rho_i \otimes \ketbra{\psi_i}{\psi_i},
\ee
with $\rho_B=\frac{\openone}{2}$. The QSE of these states is $2$-dimensional \cite{JPJR14}. We have now obtained the following result.
\begin{result}
Any discordant state whose QSE is $2$-dimensional, which has the form given by Eq. (\ref{sepn=3}),  have a nonvanishing $Q^\rightarrow_2$, on the other hand, they always have $Q^\rightarrow_3=0$.
\end{result}
\begin{proof}
Any discordant state that has the form given by Eq. (\ref{sepn=3}) exhibits 1SSDI steerability since one of its reduced states is maximally mixed \cite{JD23}. Since the QSE of any such state is $2$-dimensional, it follows that it does not have the simultaneous correlations in three MUBs quantified by $Q^\rightarrow_3$.
\end{proof}
An example of the state that satisfies the property of the above result is the Bell-diagonal state with $c_1=c_2=1/2$ and $c_3=0$. Before we express this state $\tau''$ in the form given by Eq. (\ref{sepn=3}),
 define the set ${\mathcal{W}=\{\ket{0,0},\ket{\frac{2\pi}{3},0},\ket{\frac{2\pi}{3},\pi}\}}$ with the state given by Eq. (\ref{eq:desc_sep_1}).
Then the above-mentioned Bell-diagonal state $\tau''$ can be expressed as \cite{BPP15}
    \begin{equation} \label{eq:rk3_max_opt}
    \tau'' = \frac{1}{3}\sum_{i=1}^3{W_k\otimes W_k} \,,
    \end{equation}
with ${W_k=\ket{w_k}\bra{w_k}}$ and ${\ket{w_k}\in\mathcal{W}}$, up to local unitary transformations. It has a nonzero discord value as $D^{\rightarrow}(\tau'') = 0.3113$  \cite{BPP15}. 
 For this class of state, Eqs. (\ref{Q2X}) and (\ref{Q3X}) with $c_2=1/2$ and $c_3=0$ imply operational quantification of simultaneous correlations in two and three MUBs via the Schr\"{o}dinger strengths as given by ${\rm SS}_{2}(\tau'')=1/2$ and ${\rm SS}_{3}(\tau'')=0$, respectively.

We consider another class of separable states (\ref{sep2q}) with $n=3$, which have nonmaximally mixed marginals for both subsystems and have a nonzero discord. These discordant states are considered in Ref. \cite{Gio13} to show that such discordant states are useless for the quantum communication task of remote state preparation  as proposed in Ref. \cite{DLM+12}.  Consider the following one of these states:
\ba \label{ssepQ=0}
\rho_{AB}&=&\frac{1}{4} \Big( \openone \otimes \openone + 0.4 \sigma_1 \otimes \openone +0.4 (\openone \otimes \sigma_1 - \openone \otimes \sigma_3)  \nonumber \\  
&+&0.2 \sigma_3 \otimes \sigma_3  \Big).
\ea
The above state has a nonzero discord value as $D^{\rightarrow}(\rho_{AB}) \approx 2.6 \times 10^{-2}$  \cite{Gio13}. 
In Ref. \cite{JKC+24}, it has been observed that the above state exhibits 1SSDI steerability in the two-setting scenario, but it does not have a nonvanishing Schr\"{o}dinger strength, ${\rm SS}_{2}$.
On the other hand, it has a nonzero $Q^\rightarrow_2$ since correlation in more than one basis is required to demonstrate 1SSDI steerability \cite{JKC+24}.
We have now obtained the following result.
\begin{result}
There are discordant states with a nonvanishing $Q^\rightarrow_2$ which do not have operational quantification of simultaneous correlations in two MUBs via Schr\"{o}dinger strength ${\rm SS}_{2}$.
\end{result}
 From the above result, it follows that the Schr\"{o}dinger strength ${\rm SS}_{2}$ only provides a sufficient quantification of the 1SSDI steerability. On the other hand, from the result in Cor. \ref{corQ21SSDI}, it follows that $Q_2$ can be used to provide a necessary and sufficient quantification of the 1SSDI steerability.
 Note that the state (\ref{eq:rk3_max_opt}) exhibits complete steering in the QSE formalism since Alice's reduced state is maximally mixed, whereas the state (\ref{ssepQ=0}) exhibits incomplete steering. Thus, 1SSDI steerability of the state (\ref{ssepQ=0}) not quantifiable by ${\rm SS}_{2}$ reflects its less steerability over that of the state (\ref{eq:rk3_max_opt}) by having incomplete steering in its QSE.  On the other hand, for any separable state of the form (\ref{sepn=4}), the QSE always has complete steering because it is $3$-dimensional.  
 
In sum, we have obtained a different classification of two-qubit discordant  states as stated below. 
\begin{itemize}
\item Discordant states that do not have a nonzero $Q^\rightarrow_2$. Such discordant states cannot be used to demonstrate 1SSDI steerability.
\item Discordant states that have a nonzero $Q^\rightarrow_2$, but $Q^\rightarrow_3=0$. Such states exhibit  1SSDI steerability and can be further classified into whether it has a nonzero ${\rm SS}_{2}$ or not.
\item Discordant states that have a nonzero $Q^\rightarrow_2$ as well as a nonzero $Q^\rightarrow_3$. Such states exhibiting 1SSDI steerability have a nonzero ${\rm SS}_{2}$ as well as a nonzero ${\rm SS}_{3}$. 
\end{itemize}

\section{Conclusions}\label{sec5} 
In this work, we have identified operational SCMUB in discord based on the measures in Ref. \cite{WMC+14} via 1SSDI steering. First, 
for any bipartite state, we have provided operational characterization of simultaneous correlations in two MUBs in discord as exhibiting 1SSDI steering in a two-setting scenario.

Next, we have focused on a specific class of states. For a broad class of two-qubit states that belongs to the Bell-diagonal states, we have studied the relationships between the quantification of 1SSDI steerability in two and three-setting scenarios and information-theoretic measures of SCMUB defined in Ref. \cite{WMC+14}. Such relationships reveal that the quantification of 1SSDI steerability of Bell-diagonal states provides an operational characterization of the quantification of simultaneous correlations in two and three MUBs of the Bell-diagonal states based on the measures defined in  \cite{WMC+14}. 

Here, it should be noted that in Ref. \cite{JKK+18}, for the Bell-diagonal states that exhibit $1$SDI steerability, the measures of $1$SDI steerability studied in Ref. \cite{CA16} were used to provide an operational characterization of the SCMUB measures defined in Ref. \cite{GW14}. Such measures of SCMUB \cite{GW14} are different versions of the information-theoretic quantification and are stronger than the quantification studied in Ref. \cite{WMC+14} as the former may capture more simultaneous correlations in the given state. 

Finally, we have obtained a different classification of two-qubit discordant states based on the measures of SCMUB in  \cite{WMC+14}. We have also studied its relationship with the quantification of 1SSDI steerability.
 In the case of two-qubit states, the QSE formulation has been used to provide a geometric understanding of our results on the operational quantification of SCMUB.

We solve the following open question mentioned in \cite{WMC+14} on quantifying SCMUB in discord: Which are the free states/operations in the resource theory of SCMUB in discord? We have identified that the discordant states that do not have global coherence cannot be used to demonstrate the SCMUB. This identification of free states implies that the resource theory of SCMUB is different from the resource theory of discord. And, the free operations are any local operation of the form $\Phi_A \otimes \Phi_B$, where $\Phi$ is any CPTP map. However, such free operations are more restrictive than that of entanglement theory.  

We have identified the SCMUB as a quantitative resource for 1SSDI steering in discord. Therefore, studying the quantum key distribution protocols based on quantitative 1SSDI steering in discord would be relevant. It would be interesting to identify the role of quantifying SCMUB in discord or 1SSDI steering for quantum communication tasks such as noisy teleportation, noisy one-way quantum computations, and quantum metrology. It would also be interesting to develop a resource-theoretic characterization of quantifying SCMUB in discord or 1SSDI steering such as distillation of the resource.

\section*{Acknowledgements}

CJ thanks Prof. Wei-Min Zhang for the valuable discussions. 
This work was supported by the National Science
and Technology Council (formerly the Ministry of Science and
Technology), Taiwan  (Grant No. MOST 111-2811-M-006-040-MY2), the National Science and Technology Council, the Ministry of Education (Higher Education Sprout Project NTU-113L104022-1), and the National Center for Theoretical Sciences of Taiwan.
DD acknowledges the Royal Society, United Kingdom for the support through the Newton International Fellowship (NIF$\backslash$R1$\backslash$212007). DD also acknowledges financial support from EPSRC (Engineering \& Physical Sciences Research Council, United Kingdom) Grant No. EP/X009467/1 and STFC (Science and Technology Facilities Council, United Kingdom) Grant No. ST/W006227/1.
 H.- Y. K. is supported by the Ministry of Science and Technology, Taiwan, (with grant number MOST 112-2112-M-003- 020-MY3), and Higher Education Sprout Project of National Taiwan Normal University (NTNU). 
H.-S.G. acknowledges support from the National Science and Technology Council, Taiwan under Grants No. NSTC 113-2112-M-002-022-MY3, No. NSTC 113-2119-M-002 -021 and No. NSTC 114-2119-M-002-017-MY3, from the US Air Force Office of Scientific Research under Award Number FA2386-23-1-4052 and from the National Taiwan University under Grants No. NTU-CC114L895004 and No. NTU-CC113L891604. H.-S.G. is also grateful for the support from the “Center for Advanced Computing and Imaging in Biomedicine (NTU-113L900702)” through The Featured Areas Research Center Program within the framework of the Higher Education Sprout Project by the Ministry of Education (MOE), Taiwan, the support from Taiwan Semiconductor Research Institute (TSRI) through the Joint Developed Project (JDP), and the support from the Physics Division, National Center for Theoretical Sciences, Taiwan.
   
\bibliography{scmub}

\begin{thebibliography}{56}%
\makeatletter
\providecommand \@ifxundefined [1]{%
 \@ifx{#1\undefined}
}%
\providecommand \@ifnum [1]{%
 \ifnum #1\expandafter \@firstoftwo
 \else \expandafter \@secondoftwo
 \fi
}%
\providecommand \@ifx [1]{%
 \ifx #1\expandafter \@firstoftwo
 \else \expandafter \@secondoftwo
 \fi
}%
\providecommand \natexlab [1]{#1}%
\providecommand \enquote  [1]{``#1''}%
\providecommand \bibnamefont  [1]{#1}%
\providecommand \bibfnamefont [1]{#1}%
\providecommand \citenamefont [1]{#1}%
\providecommand \href@noop [0]{\@secondoftwo}%
\providecommand \href [0]{\begingroup \@sanitize@url \@href}%
\providecommand \@href[1]{\@@startlink{#1}\@@href}%
\providecommand \@@href[1]{\endgroup#1\@@endlink}%
\providecommand \@sanitize@url [0]{\catcode `\\12\catcode `\$12\catcode `\&12\catcode `\#12\catcode `\^12\catcode `\_12\catcode `\%12\relax}%
\providecommand \@@startlink[1]{}%
\providecommand \@@endlink[0]{}%
\providecommand \url  [0]{\begingroup\@sanitize@url \@url }%
\providecommand \@url [1]{\endgroup\@href {#1}{\urlprefix }}%
\providecommand \urlprefix  [0]{URL }%
\providecommand \Eprint [0]{\href }%
\providecommand \doibase [0]{https://doi.org/}%
\providecommand \selectlanguage [0]{\@gobble}%
\providecommand \bibinfo  [0]{\@secondoftwo}%
\providecommand \bibfield  [0]{\@secondoftwo}%
\providecommand \translation [1]{[#1]}%
\providecommand \BibitemOpen [0]{}%
\providecommand \bibitemStop [0]{}%
\providecommand \bibitemNoStop [0]{.\EOS\space}%
\providecommand \EOS [0]{\spacefactor3000\relax}%
\providecommand \BibitemShut  [1]{\csname bibitem#1\endcsname}%
\let\auto@bib@innerbib\@empty
\bibitem [{\citenamefont {Tavakoli}\ \emph {et~al.}(2021)\citenamefont {Tavakoli}, \citenamefont {Farkas}, \citenamefont {Rosset}, \citenamefont {Bancal},\ and\ \citenamefont {Kaniewski}}]{TFR+21}%
  \BibitemOpen
  \bibfield  {author} {\bibinfo {author} {\bibfnamefont {A.}~\bibnamefont {Tavakoli}}, \bibinfo {author} {\bibfnamefont {M.}~\bibnamefont {Farkas}}, \bibinfo {author} {\bibfnamefont {D.}~\bibnamefont {Rosset}}, \bibinfo {author} {\bibfnamefont {J.-D.}\ \bibnamefont {Bancal}},\ and\ \bibinfo {author} {\bibfnamefont {J.}~\bibnamefont {Kaniewski}},\ }\bibfield  {title} {\bibinfo {title} {{Mutually unbiased bases and symmetric informationally complete measurements in Bell experiments}},\ }\href {https://doi.org/10.1126/sciadv.abc3847} {\bibfield  {journal} {\bibinfo  {journal} {Sci. Adv.}\ }\textbf {\bibinfo {volume} {7}},\ \bibinfo {pages} {abc3847} (\bibinfo {year} {2021})}\BibitemShut {NoStop}%
\bibitem [{\citenamefont {Skrzypczyk}\ \emph {et~al.}(2014)\citenamefont {Skrzypczyk}, \citenamefont {Navascu\'es},\ and\ \citenamefont {Cavalcanti}}]{SNC14}%
  \BibitemOpen
  \bibfield  {author} {\bibinfo {author} {\bibfnamefont {P.}~\bibnamefont {Skrzypczyk}}, \bibinfo {author} {\bibfnamefont {M.}~\bibnamefont {Navascu\'es}},\ and\ \bibinfo {author} {\bibfnamefont {D.}~\bibnamefont {Cavalcanti}},\ }\bibfield  {title} {\bibinfo {title} {Quantifying einstein-podolsky-rosen steering},\ }\href {https://doi.org/10.1103/PhysRevLett.112.180404} {\bibfield  {journal} {\bibinfo  {journal} {Phys. Rev. Lett.}\ }\textbf {\bibinfo {volume} {112}},\ \bibinfo {pages} {180404} (\bibinfo {year} {2014})}\BibitemShut {NoStop}%
\bibitem [{\citenamefont {Schrodinger}(1935)}]{Sch35}%
  \BibitemOpen
  \bibfield  {author} {\bibinfo {author} {\bibfnamefont {E.}~\bibnamefont {Schrodinger}},\ }\bibfield  {title} {\bibinfo {title} {Discussion of probability relations between separated systems},\ }\href {https://doi.org/10.1017/S0305004100013554} {\bibfield  {journal} {\bibinfo  {journal} {Math. Proc. Cambridge Philos. Soc.}\ }\textbf {\bibinfo {volume} {31}},\ \bibinfo {pages} {555} (\bibinfo {year} {1935})}\BibitemShut {NoStop}%
\bibitem [{\citenamefont {Jebarathinam}\ \emph {et~al.}(2018)\citenamefont {Jebarathinam}, \citenamefont {Khan}, \citenamefont {Kanjilal},\ and\ \citenamefont {Home}}]{JKK+18}%
  \BibitemOpen
  \bibfield  {author} {\bibinfo {author} {\bibfnamefont {C.}~\bibnamefont {Jebarathinam}}, \bibinfo {author} {\bibfnamefont {A.}~\bibnamefont {Khan}}, \bibinfo {author} {\bibfnamefont {S.}~\bibnamefont {Kanjilal}},\ and\ \bibinfo {author} {\bibfnamefont {D.}~\bibnamefont {Home}},\ }\bibfield  {title} {\bibinfo {title} {Revealing the quantitative relation between simultaneous correlations in complementary bases and quantum steering for two-qubit bell diagonal states},\ }\href {https://doi.org/10.1103/PhysRevA.98.042306} {\bibfield  {journal} {\bibinfo  {journal} {Phys. Rev. A}\ }\textbf {\bibinfo {volume} {98}},\ \bibinfo {pages} {042306} (\bibinfo {year} {2018})}\BibitemShut {NoStop}%
\bibitem [{\citenamefont {Wiseman}\ \emph {et~al.}(2007)\citenamefont {Wiseman}, \citenamefont {Jones},\ and\ \citenamefont {Doherty}}]{WJD07}%
  \BibitemOpen
  \bibfield  {author} {\bibinfo {author} {\bibfnamefont {H.~M.}\ \bibnamefont {Wiseman}}, \bibinfo {author} {\bibfnamefont {S.~J.}\ \bibnamefont {Jones}},\ and\ \bibinfo {author} {\bibfnamefont {A.~C.}\ \bibnamefont {Doherty}},\ }\bibfield  {title} {\bibinfo {title} {Steering, entanglement, nonlocality, and the einstein-podolsky-rosen paradox},\ }\href {https://doi.org/10.1103/PhysRevLett.98.140402} {\bibfield  {journal} {\bibinfo  {journal} {Phys. Rev. Lett.}\ }\textbf {\bibinfo {volume} {98}},\ \bibinfo {pages} {140402} (\bibinfo {year} {2007})}\BibitemShut {NoStop}%
\bibitem [{\citenamefont {Uola}\ \emph {et~al.}(2015)\citenamefont {Uola}, \citenamefont {Budroni}, \citenamefont {G\"uhne},\ and\ \citenamefont {Pellonp\"a\"a}}]{UBG+15}%
  \BibitemOpen
  \bibfield  {author} {\bibinfo {author} {\bibfnamefont {R.}~\bibnamefont {Uola}}, \bibinfo {author} {\bibfnamefont {C.}~\bibnamefont {Budroni}}, \bibinfo {author} {\bibfnamefont {O.}~\bibnamefont {G\"uhne}},\ and\ \bibinfo {author} {\bibfnamefont {J.-P.}\ \bibnamefont {Pellonp\"a\"a}},\ }\bibfield  {title} {\bibinfo {title} {One-to-one mapping between steering and joint measurability problems},\ }\href {https://doi.org/10.1103/PhysRevLett.115.230402} {\bibfield  {journal} {\bibinfo  {journal} {Phys. Rev. Lett.}\ }\textbf {\bibinfo {volume} {115}},\ \bibinfo {pages} {230402} (\bibinfo {year} {2015})}\BibitemShut {NoStop}%
\bibitem [{\citenamefont {Ku}\ \emph {et~al.}(2018)\citenamefont {Ku}, \citenamefont {Chen}, \citenamefont {Budroni}, \citenamefont {Miranowicz}, \citenamefont {Chen},\ and\ \citenamefont {Nori}}]{KCB+18}%
  \BibitemOpen
  \bibfield  {author} {\bibinfo {author} {\bibfnamefont {H.-Y.}\ \bibnamefont {Ku}}, \bibinfo {author} {\bibfnamefont {S.-L.}\ \bibnamefont {Chen}}, \bibinfo {author} {\bibfnamefont {C.}~\bibnamefont {Budroni}}, \bibinfo {author} {\bibfnamefont {A.}~\bibnamefont {Miranowicz}}, \bibinfo {author} {\bibfnamefont {Y.-N.}\ \bibnamefont {Chen}},\ and\ \bibinfo {author} {\bibfnamefont {F.}~\bibnamefont {Nori}},\ }\bibfield  {title} {\bibinfo {title} {Einstein-podolsky-rosen steering: Its geometric quantification and witness},\ }\href {https://doi.org/10.1103/PhysRevA.97.022338} {\bibfield  {journal} {\bibinfo  {journal} {Phys. Rev. A}\ }\textbf {\bibinfo {volume} {97}},\ \bibinfo {pages} {022338} (\bibinfo {year} {2018})}\BibitemShut {NoStop}%
\bibitem [{\citenamefont {Das}\ \emph {et~al.}(2018{\natexlab{a}})\citenamefont {Das}, \citenamefont {Datta}, \citenamefont {Jebaratnam},\ and\ \citenamefont {Majumdar}}]{DDJ+18}%
  \BibitemOpen
  \bibfield  {author} {\bibinfo {author} {\bibfnamefont {D.}~\bibnamefont {Das}}, \bibinfo {author} {\bibfnamefont {S.}~\bibnamefont {Datta}}, \bibinfo {author} {\bibfnamefont {C.}~\bibnamefont {Jebaratnam}},\ and\ \bibinfo {author} {\bibfnamefont {A.~S.}\ \bibnamefont {Majumdar}},\ }\bibfield  {title} {\bibinfo {title} {Cost of einstein-podolsky-rosen steering in the context of extremal boxes},\ }\href {https://doi.org/10.1103/PhysRevA.97.022110} {\bibfield  {journal} {\bibinfo  {journal} {Phys. Rev. A}\ }\textbf {\bibinfo {volume} {97}},\ \bibinfo {pages} {022110} (\bibinfo {year} {2018}{\natexlab{a}})}\BibitemShut {NoStop}%
\bibitem [{\citenamefont {Ku}\ \emph {et~al.}(2022)\citenamefont {Ku}, \citenamefont {Hsieh}, \citenamefont {Chen}, \citenamefont {Chen},\ and\ \citenamefont {Budroni}}]{KHCC+22}%
  \BibitemOpen
  \bibfield  {author} {\bibinfo {author} {\bibfnamefont {H.-Y.}\ \bibnamefont {Ku}}, \bibinfo {author} {\bibfnamefont {C.-Y.}\ \bibnamefont {Hsieh}}, \bibinfo {author} {\bibfnamefont {S.-L.}\ \bibnamefont {Chen}}, \bibinfo {author} {\bibfnamefont {Y.-N.}\ \bibnamefont {Chen}},\ and\ \bibinfo {author} {\bibfnamefont {C.}~\bibnamefont {Budroni}},\ }\bibfield  {title} {\bibinfo {title} {Complete classification of steerability under local filters and its relation with measurement incompatibility},\ }\bibfield  {journal} {\bibinfo  {journal} {Nature Communications}\ }\textbf {\bibinfo {volume} {13}},\ \href {https://doi.org/10.1038/s41467-022-32466-y} {10.1038/s41467-022-32466-y} (\bibinfo {year} {2022})\BibitemShut {NoStop}%
\bibitem [{\citenamefont {Wang}\ \emph {et~al.}(2024)\citenamefont {Wang}, \citenamefont {Ku}, \citenamefont {Lin},\ and\ \citenamefont {Chen}}]{Wang2024}%
  \BibitemOpen
  \bibfield  {author} {\bibinfo {author} {\bibfnamefont {H.-M.}\ \bibnamefont {Wang}}, \bibinfo {author} {\bibfnamefont {H.-Y.}\ \bibnamefont {Ku}}, \bibinfo {author} {\bibfnamefont {J.-Y.}\ \bibnamefont {Lin}},\ and\ \bibinfo {author} {\bibfnamefont {H.-B.}\ \bibnamefont {Chen}},\ }\bibfield  {title} {\bibinfo {title} {Deep learning the hierarchy of steering measurement settings of qubit-pair states},\ }\bibfield  {journal} {\bibinfo  {journal} {Communications Physics}\ }\textbf {\bibinfo {volume} {7}},\ \href {https://doi.org/10.1038/s42005-024-01563-3} {10.1038/s42005-024-01563-3} (\bibinfo {year} {2024})\BibitemShut {NoStop}%
\bibitem [{\citenamefont {{Lee}}\ \emph {et~al.}(2023)\citenamefont {{Lee}}, \citenamefont {{Lin}}, \citenamefont {{Lemr}}, \citenamefont {{{\v{C}}ernoch}}, \citenamefont {{Miranowicz}}, \citenamefont {{Nori}}, \citenamefont {{Ku}},\ and\ \citenamefont {{Chen}}}]{LLL+24}%
  \BibitemOpen
  \bibfield  {author} {\bibinfo {author} {\bibfnamefont {K.-Y.}\ \bibnamefont {{Lee}}}, \bibinfo {author} {\bibfnamefont {J.-D.}\ \bibnamefont {{Lin}}}, \bibinfo {author} {\bibfnamefont {K.}~\bibnamefont {{Lemr}}}, \bibinfo {author} {\bibfnamefont {A.}~\bibnamefont {{{\v{C}}ernoch}}}, \bibinfo {author} {\bibfnamefont {A.}~\bibnamefont {{Miranowicz}}}, \bibinfo {author} {\bibfnamefont {F.}~\bibnamefont {{Nori}}}, \bibinfo {author} {\bibfnamefont {H.-Y.}\ \bibnamefont {{Ku}}},\ and\ \bibinfo {author} {\bibfnamefont {Y.-N.}\ \bibnamefont {{Chen}}},\ }\href {https://arxiv.org/abs/2312.01055} {\bibinfo {title} {{Unveiling quantum steering by quantum-classical uncertainty complementarity}}} (\bibinfo {year} {2023}),\ \Eprint {https://arxiv.org/abs/2312.01055} {arXiv:2312.01055 [quant-ph]} \BibitemShut {NoStop}%
\bibitem [{\citenamefont {Ku}\ \emph {et~al.}(2024)\citenamefont {Ku}, \citenamefont {Hsieh},\ and\ \citenamefont {Budroni}}]{KHB24}%
  \BibitemOpen
  \bibfield  {author} {\bibinfo {author} {\bibfnamefont {H.-Y.}\ \bibnamefont {Ku}}, \bibinfo {author} {\bibfnamefont {C.-Y.}\ \bibnamefont {Hsieh}},\ and\ \bibinfo {author} {\bibfnamefont {C.}~\bibnamefont {Budroni}},\ }\href {https://arxiv.org/abs/2308.02252} {\bibinfo {title} {Measurement incompatibility cannot be stochastically distilled}} (\bibinfo {year} {2024}),\ \Eprint {https://arxiv.org/abs/2308.02252} {arXiv:2308.02252 [quant-ph]} \BibitemShut {NoStop}%
\bibitem [{\citenamefont {Branciard}\ \emph {et~al.}(2012)\citenamefont {Branciard}, \citenamefont {Cavalcanti}, \citenamefont {Walborn}, \citenamefont {Scarani},\ and\ \citenamefont {Wiseman}}]{BCW+12}%
  \BibitemOpen
  \bibfield  {author} {\bibinfo {author} {\bibfnamefont {C.}~\bibnamefont {Branciard}}, \bibinfo {author} {\bibfnamefont {E.~G.}\ \bibnamefont {Cavalcanti}}, \bibinfo {author} {\bibfnamefont {S.~P.}\ \bibnamefont {Walborn}}, \bibinfo {author} {\bibfnamefont {V.}~\bibnamefont {Scarani}},\ and\ \bibinfo {author} {\bibfnamefont {H.~M.}\ \bibnamefont {Wiseman}},\ }\bibfield  {title} {\bibinfo {title} {One-sided device-independent quantum key distribution: Security, feasibility, and the connection with steering},\ }\href {https://doi.org/10.1103/PhysRevA.85.010301} {\bibfield  {journal} {\bibinfo  {journal} {Phys. Rev. A}\ }\textbf {\bibinfo {volume} {85}},\ \bibinfo {pages} {010301 (R)} (\bibinfo {year} {2012})}\BibitemShut {NoStop}%
\bibitem [{\citenamefont {Passaro}\ \emph {et~al.}(2015)\citenamefont {Passaro}, \citenamefont {Cavalcanti}, \citenamefont {Skrzypczyk},\ and\ \citenamefont {Acín}}]{PCS+15}%
  \BibitemOpen
  \bibfield  {author} {\bibinfo {author} {\bibfnamefont {E.}~\bibnamefont {Passaro}}, \bibinfo {author} {\bibfnamefont {D.}~\bibnamefont {Cavalcanti}}, \bibinfo {author} {\bibfnamefont {P.}~\bibnamefont {Skrzypczyk}},\ and\ \bibinfo {author} {\bibfnamefont {A.}~\bibnamefont {Acín}},\ }\bibfield  {title} {\bibinfo {title} {Optimal randomness certification in the quantum steering and prepare-and-measure scenarios},\ }\href {http://stacks.iop.org/1367-2630/17/i=11/a=113010} {\bibfield  {journal} {\bibinfo  {journal} {New J. Phys.}\ }\textbf {\bibinfo {volume} {17}},\ \bibinfo {pages} {113010} (\bibinfo {year} {2015})}\BibitemShut {NoStop}%
\bibitem [{\citenamefont {Gallego}\ and\ \citenamefont {Aolita}(2015)}]{GA15}%
  \BibitemOpen
  \bibfield  {author} {\bibinfo {author} {\bibfnamefont {R.}~\bibnamefont {Gallego}}\ and\ \bibinfo {author} {\bibfnamefont {L.}~\bibnamefont {Aolita}},\ }\bibfield  {title} {\bibinfo {title} {Resource theory of steering},\ }\href {https://doi.org/10.1103/PhysRevX.5.041008} {\bibfield  {journal} {\bibinfo  {journal} {Phys. Rev. X}\ }\textbf {\bibinfo {volume} {5}},\ \bibinfo {pages} {041008} (\bibinfo {year} {2015})}\BibitemShut {NoStop}%
\bibitem [{\citenamefont {Lin}\ \emph {et~al.}(2021)\citenamefont {Lin}, \citenamefont {Lin}, \citenamefont {Ku}, \citenamefont {Lambert}, \citenamefont {Chen},\ and\ \citenamefont {Nori}}]{LJW+21}%
  \BibitemOpen
  \bibfield  {author} {\bibinfo {author} {\bibfnamefont {J.-D.}\ \bibnamefont {Lin}}, \bibinfo {author} {\bibfnamefont {W.-Y.}\ \bibnamefont {Lin}}, \bibinfo {author} {\bibfnamefont {H.-Y.}\ \bibnamefont {Ku}}, \bibinfo {author} {\bibfnamefont {N.}~\bibnamefont {Lambert}}, \bibinfo {author} {\bibfnamefont {Y.-N.}\ \bibnamefont {Chen}},\ and\ \bibinfo {author} {\bibfnamefont {F.}~\bibnamefont {Nori}},\ }\bibfield  {title} {\bibinfo {title} {Quantum steering as a witness of quantum scrambling},\ }\href {https://doi.org/10.1103/PhysRevA.104.022614} {\bibfield  {journal} {\bibinfo  {journal} {Phys. Rev. A}\ }\textbf {\bibinfo {volume} {104}},\ \bibinfo {pages} {022614} (\bibinfo {year} {2021})}\BibitemShut {NoStop}%
\bibitem [{\citenamefont {Sarkar}\ \emph {et~al.}(2023)\citenamefont {Sarkar}, \citenamefont {Borka\l{}a}, \citenamefont {Jebarathinam}, \citenamefont {Makuta}, \citenamefont {Saha},\ and\ \citenamefont {Augusiak}}]{SBJ+23}%
  \BibitemOpen
  \bibfield  {author} {\bibinfo {author} {\bibfnamefont {S.}~\bibnamefont {Sarkar}}, \bibinfo {author} {\bibfnamefont {J.~J.}\ \bibnamefont {Borka\l{}a}}, \bibinfo {author} {\bibfnamefont {C.}~\bibnamefont {Jebarathinam}}, \bibinfo {author} {\bibfnamefont {O.}~\bibnamefont {Makuta}}, \bibinfo {author} {\bibfnamefont {D.}~\bibnamefont {Saha}},\ and\ \bibinfo {author} {\bibfnamefont {R.}~\bibnamefont {Augusiak}},\ }\bibfield  {title} {\bibinfo {title} {Self-testing of any pure entangled state with the minimal number of measurements and optimal randomness certification in a one-sided device-independent scenario},\ }\href {https://doi.org/10.1103/PhysRevApplied.19.034038} {\bibfield  {journal} {\bibinfo  {journal} {Phys. Rev. Appl.}\ }\textbf {\bibinfo {volume} {19}},\ \bibinfo {pages} {034038} (\bibinfo {year} {2023})}\BibitemShut {NoStop}%
\bibitem [{\citenamefont {Hsieh}\ \emph {et~al.}(2023)\citenamefont {Hsieh}, \citenamefont {Ku},\ and\ \citenamefont {Budroni}}]{HKB23}%
  \BibitemOpen
  \bibfield  {author} {\bibinfo {author} {\bibfnamefont {C.-Y.}\ \bibnamefont {Hsieh}}, \bibinfo {author} {\bibfnamefont {H.-Y.}\ \bibnamefont {Ku}},\ and\ \bibinfo {author} {\bibfnamefont {C.}~\bibnamefont {Budroni}},\ }\href {https://arxiv.org/abs/2309.06191} {\bibinfo {title} {Characterisation and fundamental limitations of irreversible stochastic steering distillation}} (\bibinfo {year} {2023}),\ \Eprint {https://arxiv.org/abs/2309.06191} {arXiv:2309.06191 [quant-ph]} \BibitemShut {NoStop}%
\bibitem [{\citenamefont {Jebarathinam}\ \emph {et~al.}(2023)\citenamefont {Jebarathinam}, \citenamefont {Das},\ and\ \citenamefont {Srikanth}}]{JDS_PRA23}%
  \BibitemOpen
  \bibfield  {author} {\bibinfo {author} {\bibfnamefont {C.}~\bibnamefont {Jebarathinam}}, \bibinfo {author} {\bibfnamefont {D.}~\bibnamefont {Das}},\ and\ \bibinfo {author} {\bibfnamefont {R.}~\bibnamefont {Srikanth}},\ }\bibfield  {title} {\bibinfo {title} {Asymmetric one-sided semi-device-independent steerability of quantum discordant states},\ }\href {https://doi.org/10.1103/PhysRevA.108.042211} {\bibfield  {journal} {\bibinfo  {journal} {Phys. Rev. A}\ }\textbf {\bibinfo {volume} {108}},\ \bibinfo {pages} {042211} (\bibinfo {year} {2023})}\BibitemShut {NoStop}%
\bibitem [{\citenamefont {Das}\ \emph {et~al.}(2018{\natexlab{b}})\citenamefont {Das}, \citenamefont {Bhattacharya}, \citenamefont {Datta}, \citenamefont {Roy}, \citenamefont {Jebaratnam}, \citenamefont {Majumdar},\ and\ \citenamefont {Srikanth}}]{DBD+18}%
  \BibitemOpen
  \bibfield  {author} {\bibinfo {author} {\bibfnamefont {D.}~\bibnamefont {Das}}, \bibinfo {author} {\bibfnamefont {B.}~\bibnamefont {Bhattacharya}}, \bibinfo {author} {\bibfnamefont {C.}~\bibnamefont {Datta}}, \bibinfo {author} {\bibfnamefont {A.}~\bibnamefont {Roy}}, \bibinfo {author} {\bibfnamefont {C.}~\bibnamefont {Jebaratnam}}, \bibinfo {author} {\bibfnamefont {A.~S.}\ \bibnamefont {Majumdar}},\ and\ \bibinfo {author} {\bibfnamefont {R.}~\bibnamefont {Srikanth}},\ }\bibfield  {title} {\bibinfo {title} {Operational characterization of quantumness of unsteerable bipartite states},\ }\href {https://doi.org/10.1103/PhysRevA.97.062335} {\bibfield  {journal} {\bibinfo  {journal} {Phys. Rev. A}\ }\textbf {\bibinfo {volume} {97}},\ \bibinfo {pages} {062335} (\bibinfo {year} {2018}{\natexlab{b}})}\BibitemShut {NoStop}%
\bibitem [{\citenamefont {Das}\ \emph {et~al.}(2018{\natexlab{c}})\citenamefont {Das}, \citenamefont {Jebaratnam}, \citenamefont {Bhattacharya}, \citenamefont {Mukherjee}, \citenamefont {Bhattacharya},\ and\ \citenamefont {Roy}}]{DJB+18}%
  \BibitemOpen
  \bibfield  {author} {\bibinfo {author} {\bibfnamefont {D.}~\bibnamefont {Das}}, \bibinfo {author} {\bibfnamefont {C.}~\bibnamefont {Jebaratnam}}, \bibinfo {author} {\bibfnamefont {B.}~\bibnamefont {Bhattacharya}}, \bibinfo {author} {\bibfnamefont {A.}~\bibnamefont {Mukherjee}}, \bibinfo {author} {\bibfnamefont {S.~S.}\ \bibnamefont {Bhattacharya}},\ and\ \bibinfo {author} {\bibfnamefont {A.}~\bibnamefont {Roy}},\ }\bibfield  {title} {\bibinfo {title} {Characterization of the quantumness of unsteerable tripartite correlations},\ }\href {https://doi.org/https://doi.org/10.1016/j.aop.2018.08.013} {\bibfield  {journal} {\bibinfo  {journal} {Annals of Physics}\ }\textbf {\bibinfo {volume} {398}},\ \bibinfo {pages} {55 } (\bibinfo {year} {2018}{\natexlab{c}})}\BibitemShut {NoStop}%
\bibitem [{\citenamefont {Ollivier}\ and\ \citenamefont {Zurek}(2001)}]{OZ01}%
  \BibitemOpen
  \bibfield  {author} {\bibinfo {author} {\bibfnamefont {H.}~\bibnamefont {Ollivier}}\ and\ \bibinfo {author} {\bibfnamefont {W.~H.}\ \bibnamefont {Zurek}},\ }\bibfield  {title} {\bibinfo {title} {Quantum discord: A measure of the quantumness of correlations},\ }\href {https://doi.org/10.1103/PhysRevLett.88.017901} {\bibfield  {journal} {\bibinfo  {journal} {Phys. Rev. Lett.}\ }\textbf {\bibinfo {volume} {88}},\ \bibinfo {pages} {017901} (\bibinfo {year} {2001})}\BibitemShut {NoStop}%
\bibitem [{\citenamefont {Henderson}\ and\ \citenamefont {Vedral}(2001)}]{HV01}%
  \BibitemOpen
  \bibfield  {author} {\bibinfo {author} {\bibfnamefont {L.}~\bibnamefont {Henderson}}\ and\ \bibinfo {author} {\bibfnamefont {V.}~\bibnamefont {Vedral}},\ }\bibfield  {title} {\bibinfo {title} {Classical, quantum and total correlations},\ }\href {http://stacks.iop.org/0305-4470/34/i=35/a=315} {\bibfield  {journal} {\bibinfo  {journal} {J. Phys. A}\ }\textbf {\bibinfo {volume} {34}},\ \bibinfo {pages} {6899} (\bibinfo {year} {2001})}\BibitemShut {NoStop}%
\bibitem [{\citenamefont {Jebarathinam}\ \emph {et~al.}(2024)\citenamefont {Jebarathinam}, \citenamefont {Ku}, \citenamefont {Cheng},\ and\ \citenamefont {Goan}}]{JKC+24}%
  \BibitemOpen
  \bibfield  {author} {\bibinfo {author} {\bibfnamefont {C.}~\bibnamefont {Jebarathinam}}, \bibinfo {author} {\bibfnamefont {H.-Y.}\ \bibnamefont {Ku}}, \bibinfo {author} {\bibfnamefont {H.-C.}\ \bibnamefont {Cheng}},\ and\ \bibinfo {author} {\bibfnamefont {H.-S.}\ \bibnamefont {Goan}},\ }\href {https://arxiv.org/abs/2410.04430} {\bibinfo {title} {The aspect of bipartite coherence in quantum discord to semi-device-independent nonlocality and its implication for quantum information processing}} (\bibinfo {year} {2024}),\ \Eprint {https://arxiv.org/abs/2410.04430} {arXiv:2410.04430 [quant-ph]} \BibitemShut {NoStop}%
\bibitem [{\citenamefont {Jebarathinam}\ \emph {et~al.}(2019)\citenamefont {Jebarathinam}, \citenamefont {Das}, \citenamefont {Kanjilal}, \citenamefont {Srikanth}, \citenamefont {Sarkar}, \citenamefont {Chattopadhyay},\ and\ \citenamefont {Majumdar}}]{JDK+19}%
  \BibitemOpen
  \bibfield  {author} {\bibinfo {author} {\bibfnamefont {C.}~\bibnamefont {Jebarathinam}}, \bibinfo {author} {\bibfnamefont {D.}~\bibnamefont {Das}}, \bibinfo {author} {\bibfnamefont {S.}~\bibnamefont {Kanjilal}}, \bibinfo {author} {\bibfnamefont {R.}~\bibnamefont {Srikanth}}, \bibinfo {author} {\bibfnamefont {D.}~\bibnamefont {Sarkar}}, \bibinfo {author} {\bibfnamefont {I.}~\bibnamefont {Chattopadhyay}},\ and\ \bibinfo {author} {\bibfnamefont {A.~S.}\ \bibnamefont {Majumdar}},\ }\bibfield  {title} {\bibinfo {title} {Superunsteerability as a quantifiable resource for random access codes assisted by bell-diagonal states},\ }\href {https://doi.org/10.1103/PhysRevA.100.012344} {\bibfield  {journal} {\bibinfo  {journal} {Phys. Rev. A}\ }\textbf {\bibinfo {volume} {100}},\ \bibinfo {pages} {012344} (\bibinfo {year} {2019})}\BibitemShut {NoStop}%
\bibitem [{\citenamefont {Jebarathinam}\ and\ \citenamefont {Das}(2023)}]{JD23}%
  \BibitemOpen
  \bibfield  {author} {\bibinfo {author} {\bibfnamefont {C.}~\bibnamefont {Jebarathinam}}\ and\ \bibinfo {author} {\bibfnamefont {D.}~\bibnamefont {Das}},\ }\bibfield  {title} {\bibinfo {title} {Certifying quantumness beyond steering and nonlocality and its implications on quantum information processing},\ }\href {https://doi.org/10.26421/QIC23.5-6-2} {\bibfield  {journal} {\bibinfo  {journal} {Quantum Information and Computation}\ }\textbf {\bibinfo {volume} {23}},\ \bibinfo {pages} {379} (\bibinfo {year} {2023})}\BibitemShut {NoStop}%
\bibitem [{\citenamefont {Dakic}\ \emph {et~al.}(2012)\citenamefont {Dakic}, \citenamefont {Lipp}, \citenamefont {song Ma}, \citenamefont {Ringbauer}, \citenamefont {Kropatschek}, \citenamefont {Barz}, \citenamefont {Paterek}, \citenamefont {Vedral}, \citenamefont {Zeilinger}, \citenamefont {Brukner},\ and\ \citenamefont {Walther}}]{DLM+12}%
  \BibitemOpen
  \bibfield  {author} {\bibinfo {author} {\bibfnamefont {B.}~\bibnamefont {Dakic}}, \bibinfo {author} {\bibfnamefont {Y.~O.}\ \bibnamefont {Lipp}}, \bibinfo {author} {\bibfnamefont {X.}~\bibnamefont {song Ma}}, \bibinfo {author} {\bibfnamefont {M.}~\bibnamefont {Ringbauer}}, \bibinfo {author} {\bibfnamefont {S.}~\bibnamefont {Kropatschek}}, \bibinfo {author} {\bibfnamefont {S.}~\bibnamefont {Barz}}, \bibinfo {author} {\bibfnamefont {T.}~\bibnamefont {Paterek}}, \bibinfo {author} {\bibfnamefont {V.}~\bibnamefont {Vedral}}, \bibinfo {author} {\bibfnamefont {A.}~\bibnamefont {Zeilinger}}, \bibinfo {author} {\bibfnamefont {C.}~\bibnamefont {Brukner}},\ and\ \bibinfo {author} {\bibfnamefont {P.}~\bibnamefont {Walther}},\ }\bibfield  {title} {\bibinfo {title} {Quantum discord as resource for remote state preparation},\ }\href {https://doi.org/10.1038/nphys2377} {\bibfield  {journal} {\bibinfo  {journal} {Nature Physics}\ }\textbf {\bibinfo {volume} {8}},\ \bibinfo {pages} {666 } (\bibinfo {year}
  {2012})}\BibitemShut {NoStop}%
\bibitem [{\citenamefont {Wu}\ \emph {et~al.}(2014)\citenamefont {Wu}, \citenamefont {Ma}, \citenamefont {Chen},\ and\ \citenamefont {Yu}}]{WMC+14}%
  \BibitemOpen
  \bibfield  {author} {\bibinfo {author} {\bibfnamefont {S.}~\bibnamefont {Wu}}, \bibinfo {author} {\bibfnamefont {Z.}~\bibnamefont {Ma}}, \bibinfo {author} {\bibfnamefont {Z.}~\bibnamefont {Chen}},\ and\ \bibinfo {author} {\bibfnamefont {S.}~\bibnamefont {Yu}},\ }\bibfield  {title} {\bibinfo {title} {Reveal quantum correlation in complementary bases},\ }\href {https://doi.org/10.1038/srep04036} {\bibfield  {journal} {\bibinfo  {journal} {Scientific Reports}\ }\textbf {\bibinfo {volume} {4}},\ \bibinfo {pages} {4036} (\bibinfo {year} {2014})}\BibitemShut {NoStop}%
\bibitem [{\citenamefont {Barnett}\ and\ \citenamefont {Phoenix}(1989)}]{qmi}%
  \BibitemOpen
  \bibfield  {author} {\bibinfo {author} {\bibfnamefont {S.~M.}\ \bibnamefont {Barnett}}\ and\ \bibinfo {author} {\bibfnamefont {S.~J.~D.}\ \bibnamefont {Phoenix}},\ }\bibfield  {title} {\bibinfo {title} {Entropy as a measure of quantum optical correlation},\ }\href {https://doi.org/10.1103/PhysRevA.40.2404} {\bibfield  {journal} {\bibinfo  {journal} {Phys. Rev. A}\ }\textbf {\bibinfo {volume} {40}},\ \bibinfo {pages} {2404} (\bibinfo {year} {1989})}\BibitemShut {NoStop}%
\bibitem [{\citenamefont {Cerf}\ and\ \citenamefont {Adami}(1997)}]{Cerf}%
  \BibitemOpen
  \bibfield  {author} {\bibinfo {author} {\bibfnamefont {N.~J.}\ \bibnamefont {Cerf}}\ and\ \bibinfo {author} {\bibfnamefont {C.}~\bibnamefont {Adami}},\ }\bibfield  {title} {\bibinfo {title} {Negative entropy and information in quantum mechanics},\ }\href {https://doi.org/10.1103/PhysRevLett.79.5194} {\bibfield  {journal} {\bibinfo  {journal} {Phys. Rev. Lett.}\ }\textbf {\bibinfo {volume} {79}},\ \bibinfo {pages} {5194} (\bibinfo {year} {1997})}\BibitemShut {NoStop}%
\bibitem [{\citenamefont {Schumacher}\ and\ \citenamefont {Nielsen}(1996)}]{SN96}%
  \BibitemOpen
  \bibfield  {author} {\bibinfo {author} {\bibfnamefont {B.}~\bibnamefont {Schumacher}}\ and\ \bibinfo {author} {\bibfnamefont {M.~A.}\ \bibnamefont {Nielsen}},\ }\bibfield  {title} {\bibinfo {title} {Quantum data processing and error correction},\ }\href {https://doi.org/10.1103/PhysRevA.54.2629} {\bibfield  {journal} {\bibinfo  {journal} {Phys. Rev. A}\ }\textbf {\bibinfo {volume} {54}},\ \bibinfo {pages} {2629} (\bibinfo {year} {1996})}\BibitemShut {NoStop}%
\bibitem [{\citenamefont {Groisman}\ \emph {et~al.}(2005)\citenamefont {Groisman}, \citenamefont {Popescu},\ and\ \citenamefont {Winter}}]{GROIS}%
  \BibitemOpen
  \bibfield  {author} {\bibinfo {author} {\bibfnamefont {B.}~\bibnamefont {Groisman}}, \bibinfo {author} {\bibfnamefont {S.}~\bibnamefont {Popescu}},\ and\ \bibinfo {author} {\bibfnamefont {A.}~\bibnamefont {Winter}},\ }\bibfield  {title} {\bibinfo {title} {Quantum, classical, and total amount of correlations in a quantum state},\ }\href {https://doi.org/10.1103/PhysRevA.72.032317} {\bibfield  {journal} {\bibinfo  {journal} {Phys. Rev. A}\ }\textbf {\bibinfo {volume} {72}},\ \bibinfo {pages} {032317} (\bibinfo {year} {2005})}\BibitemShut {NoStop}%
\bibitem [{\citenamefont {Bromley}\ \emph {et~al.}(2015)\citenamefont {Bromley}, \citenamefont {Cianciaruso},\ and\ \citenamefont {Adesso}}]{BCA15}%
  \BibitemOpen
  \bibfield  {author} {\bibinfo {author} {\bibfnamefont {T.~R.}\ \bibnamefont {Bromley}}, \bibinfo {author} {\bibfnamefont {M.}~\bibnamefont {Cianciaruso}},\ and\ \bibinfo {author} {\bibfnamefont {G.}~\bibnamefont {Adesso}},\ }\bibfield  {title} {\bibinfo {title} {Frozen quantum coherence},\ }\href {https://doi.org/10.1103/PhysRevLett.114.210401} {\bibfield  {journal} {\bibinfo  {journal} {Phys. Rev. Lett.}\ }\textbf {\bibinfo {volume} {114}},\ \bibinfo {pages} {210401} (\bibinfo {year} {2015})}\BibitemShut {NoStop}%
\bibitem [{\citenamefont {Streltsov}\ \emph {et~al.}(2015)\citenamefont {Streltsov}, \citenamefont {Singh}, \citenamefont {Dhar}, \citenamefont {Bera},\ and\ \citenamefont {Adesso}}]{SSD+15}%
  \BibitemOpen
  \bibfield  {author} {\bibinfo {author} {\bibfnamefont {A.}~\bibnamefont {Streltsov}}, \bibinfo {author} {\bibfnamefont {U.}~\bibnamefont {Singh}}, \bibinfo {author} {\bibfnamefont {H.~S.}\ \bibnamefont {Dhar}}, \bibinfo {author} {\bibfnamefont {M.~N.}\ \bibnamefont {Bera}},\ and\ \bibinfo {author} {\bibfnamefont {G.}~\bibnamefont {Adesso}},\ }\bibfield  {title} {\bibinfo {title} {Measuring quantum coherence with entanglement},\ }\href {https://doi.org/10.1103/PhysRevLett.115.020403} {\bibfield  {journal} {\bibinfo  {journal} {Phys. Rev. Lett.}\ }\textbf {\bibinfo {volume} {115}},\ \bibinfo {pages} {020403} (\bibinfo {year} {2015})}\BibitemShut {NoStop}%
\bibitem [{\citenamefont {Baumgratz}\ \emph {et~al.}(2014)\citenamefont {Baumgratz}, \citenamefont {Cramer},\ and\ \citenamefont {Plenio}}]{BCP14}%
  \BibitemOpen
  \bibfield  {author} {\bibinfo {author} {\bibfnamefont {T.}~\bibnamefont {Baumgratz}}, \bibinfo {author} {\bibfnamefont {M.}~\bibnamefont {Cramer}},\ and\ \bibinfo {author} {\bibfnamefont {M.~B.}\ \bibnamefont {Plenio}},\ }\bibfield  {title} {\bibinfo {title} {Quantifying coherence},\ }\href {https://doi.org/10.1103/PhysRevLett.113.140401} {\bibfield  {journal} {\bibinfo  {journal} {Phys. Rev. Lett.}\ }\textbf {\bibinfo {volume} {113}},\ \bibinfo {pages} {140401} (\bibinfo {year} {2014})}\BibitemShut {NoStop}%
\bibitem [{\citenamefont {Streltsov}\ \emph {et~al.}(2017)\citenamefont {Streltsov}, \citenamefont {Adesso},\ and\ \citenamefont {Plenio}}]{SAP17}%
  \BibitemOpen
  \bibfield  {author} {\bibinfo {author} {\bibfnamefont {A.}~\bibnamefont {Streltsov}}, \bibinfo {author} {\bibfnamefont {G.}~\bibnamefont {Adesso}},\ and\ \bibinfo {author} {\bibfnamefont {M.~B.}\ \bibnamefont {Plenio}},\ }\bibfield  {title} {\bibinfo {title} {Colloquium: Quantum coherence as a resource},\ }\href {https://doi.org/10.1103/RevModPhys.89.041003} {\bibfield  {journal} {\bibinfo  {journal} {Rev. Mod. Phys.}\ }\textbf {\bibinfo {volume} {89}},\ \bibinfo {pages} {041003} (\bibinfo {year} {2017})}\BibitemShut {NoStop}%
\bibitem [{\citenamefont {Yao}\ \emph {et~al.}(2015)\citenamefont {Yao}, \citenamefont {Xiao}, \citenamefont {Ge},\ and\ \citenamefont {Sun}}]{YXL+15}%
  \BibitemOpen
  \bibfield  {author} {\bibinfo {author} {\bibfnamefont {Y.}~\bibnamefont {Yao}}, \bibinfo {author} {\bibfnamefont {X.}~\bibnamefont {Xiao}}, \bibinfo {author} {\bibfnamefont {L.}~\bibnamefont {Ge}},\ and\ \bibinfo {author} {\bibfnamefont {C.~P.}\ \bibnamefont {Sun}},\ }\bibfield  {title} {\bibinfo {title} {Quantum coherence in multipartite systems},\ }\href {https://doi.org/10.1103/PhysRevA.92.022112} {\bibfield  {journal} {\bibinfo  {journal} {Phys. Rev. A}\ }\textbf {\bibinfo {volume} {92}},\ \bibinfo {pages} {022112} (\bibinfo {year} {2015})}\BibitemShut {NoStop}%
\bibitem [{\citenamefont {Tan}\ \emph {et~al.}(2016)\citenamefont {Tan}, \citenamefont {Kwon}, \citenamefont {Park},\ and\ \citenamefont {Jeong}}]{BKP+16}%
  \BibitemOpen
  \bibfield  {author} {\bibinfo {author} {\bibfnamefont {K.~C.}\ \bibnamefont {Tan}}, \bibinfo {author} {\bibfnamefont {H.}~\bibnamefont {Kwon}}, \bibinfo {author} {\bibfnamefont {C.-Y.}\ \bibnamefont {Park}},\ and\ \bibinfo {author} {\bibfnamefont {H.}~\bibnamefont {Jeong}},\ }\bibfield  {title} {\bibinfo {title} {Unified view of quantum correlations and quantum coherence},\ }\href {https://doi.org/10.1103/PhysRevA.94.022329} {\bibfield  {journal} {\bibinfo  {journal} {Phys. Rev. A}\ }\textbf {\bibinfo {volume} {94}},\ \bibinfo {pages} {022329} (\bibinfo {year} {2016})}\BibitemShut {NoStop}%
\bibitem [{\citenamefont {Wang}\ \emph {et~al.}(2017)\citenamefont {Wang}, \citenamefont {Yue}, \citenamefont {Yu}, \citenamefont {Gao},\ and\ \citenamefont {Qin}}]{WYY+17}%
  \BibitemOpen
  \bibfield  {author} {\bibinfo {author} {\bibfnamefont {X.-L.}\ \bibnamefont {Wang}}, \bibinfo {author} {\bibfnamefont {Q.}~\bibnamefont {Yue}}, \bibinfo {author} {\bibfnamefont {C.-H.}\ \bibnamefont {Yu}}, \bibinfo {author} {\bibfnamefont {F.}~\bibnamefont {Gao}},\ and\ \bibinfo {author} {\bibfnamefont {S.}~\bibnamefont {Qin}},\ }\bibfield  {title} {\bibinfo {title} {Relating quantum coherence and correlations with entropy-based measures},\ }\bibfield  {journal} {\bibinfo  {journal} {Scientific Reports}\ }\textbf {\bibinfo {volume} {7}},\ \href {https://doi.org/10.1038/s41598-017-09332-9} {10.1038/s41598-017-09332-9} (\bibinfo {year} {2017})\BibitemShut {NoStop}%
\bibitem [{\citenamefont {Guo}\ and\ \citenamefont {Goswami}(2017)}]{GG17}%
  \BibitemOpen
  \bibfield  {author} {\bibinfo {author} {\bibfnamefont {Y.}~\bibnamefont {Guo}}\ and\ \bibinfo {author} {\bibfnamefont {S.}~\bibnamefont {Goswami}},\ }\bibfield  {title} {\bibinfo {title} {Discordlike correlation of bipartite coherence},\ }\href {https://doi.org/10.1103/PhysRevA.95.062340} {\bibfield  {journal} {\bibinfo  {journal} {Phys. Rev. A}\ }\textbf {\bibinfo {volume} {95}},\ \bibinfo {pages} {062340} (\bibinfo {year} {2017})}\BibitemShut {NoStop}%
\bibitem [{\citenamefont {Li}\ \emph {et~al.}(2024)\citenamefont {Li}, \citenamefont {Zhang}, \citenamefont {Luo},\ and\ \citenamefont {Sun}}]{LZZ+24}%
  \BibitemOpen
  \bibfield  {author} {\bibinfo {author} {\bibfnamefont {N.}~\bibnamefont {Li}}, \bibinfo {author} {\bibfnamefont {Z.}~\bibnamefont {Zhang}}, \bibinfo {author} {\bibfnamefont {S.}~\bibnamefont {Luo}},\ and\ \bibinfo {author} {\bibfnamefont {Y.}~\bibnamefont {Sun}},\ }\bibfield  {title} {\bibinfo {title} {Characterizing bipartite states with vanishing basis-dependent correlations},\ }\href {https://doi.org/10.1103/PhysRevA.110.022418} {\bibfield  {journal} {\bibinfo  {journal} {Phys. Rev. A}\ }\textbf {\bibinfo {volume} {110}},\ \bibinfo {pages} {022418} (\bibinfo {year} {2024})}\BibitemShut {NoStop}%
\bibitem [{\citenamefont {Bellomo}\ \emph {et~al.}(2015)\citenamefont {Bellomo}, \citenamefont {Plastino},\ and\ \citenamefont {Plastino}}]{BPP15}%
  \BibitemOpen
  \bibfield  {author} {\bibinfo {author} {\bibfnamefont {G.}~\bibnamefont {Bellomo}}, \bibinfo {author} {\bibfnamefont {A.}~\bibnamefont {Plastino}},\ and\ \bibinfo {author} {\bibfnamefont {A.~R.}\ \bibnamefont {Plastino}},\ }\bibfield  {title} {\bibinfo {title} {Classical extension of quantum-correlated separable states},\ }\href {https://doi.org/10.1142/S021974991550015X} {\bibfield  {journal} {\bibinfo  {journal} {International Journal of Quantum Information}\ }\textbf {\bibinfo {volume} {13}},\ \bibinfo {pages} {1550015} (\bibinfo {year} {2015})}\BibitemShut {NoStop}%
\bibitem [{\citenamefont {Streltsov}\ \emph {et~al.}(2011)\citenamefont {Streltsov}, \citenamefont {Kampermann},\ and\ \citenamefont {Bru\ss{}}}]{SKB11}%
  \BibitemOpen
  \bibfield  {author} {\bibinfo {author} {\bibfnamefont {A.}~\bibnamefont {Streltsov}}, \bibinfo {author} {\bibfnamefont {H.}~\bibnamefont {Kampermann}},\ and\ \bibinfo {author} {\bibfnamefont {D.}~\bibnamefont {Bru\ss{}}},\ }\bibfield  {title} {\bibinfo {title} {Behavior of quantum correlations under local noise},\ }\href {https://doi.org/10.1103/PhysRevLett.107.170502} {\bibfield  {journal} {\bibinfo  {journal} {Phys. Rev. Lett.}\ }\textbf {\bibinfo {volume} {107}},\ \bibinfo {pages} {170502} (\bibinfo {year} {2011})}\BibitemShut {NoStop}%
\bibitem [{\citenamefont {Gessner}\ \emph {et~al.}(2012)\citenamefont {Gessner}, \citenamefont {Laine}, \citenamefont {Breuer},\ and\ \citenamefont {Piilo}}]{GLB+12}%
  \BibitemOpen
  \bibfield  {author} {\bibinfo {author} {\bibfnamefont {M.}~\bibnamefont {Gessner}}, \bibinfo {author} {\bibfnamefont {E.-M.}\ \bibnamefont {Laine}}, \bibinfo {author} {\bibfnamefont {H.-P.}\ \bibnamefont {Breuer}},\ and\ \bibinfo {author} {\bibfnamefont {J.}~\bibnamefont {Piilo}},\ }\bibfield  {title} {\bibinfo {title} {Correlations in quantum states and the local creation of quantum discord},\ }\href {https://doi.org/10.1103/PhysRevA.85.052122} {\bibfield  {journal} {\bibinfo  {journal} {Phys. Rev. A}\ }\textbf {\bibinfo {volume} {85}},\ \bibinfo {pages} {052122} (\bibinfo {year} {2012})}\BibitemShut {NoStop}%
\bibitem [{\citenamefont {Giorgi}(2013)}]{Gio13}%
  \BibitemOpen
  \bibfield  {author} {\bibinfo {author} {\bibfnamefont {G.~L.}\ \bibnamefont {Giorgi}},\ }\bibfield  {title} {\bibinfo {title} {Quantum discord and remote state preparation},\ }\href {https://doi.org/10.1103/PhysRevA.88.022315} {\bibfield  {journal} {\bibinfo  {journal} {Phys. Rev. A}\ }\textbf {\bibinfo {volume} {88}},\ \bibinfo {pages} {022315} (\bibinfo {year} {2013})}\BibitemShut {NoStop}%
\bibitem [{\citenamefont {Yadin}\ \emph {et~al.}(2016)\citenamefont {Yadin}, \citenamefont {Ma}, \citenamefont {Girolami}, \citenamefont {Gu},\ and\ \citenamefont {Vedral}}]{YMG+16}%
  \BibitemOpen
  \bibfield  {author} {\bibinfo {author} {\bibfnamefont {B.}~\bibnamefont {Yadin}}, \bibinfo {author} {\bibfnamefont {J.}~\bibnamefont {Ma}}, \bibinfo {author} {\bibfnamefont {D.}~\bibnamefont {Girolami}}, \bibinfo {author} {\bibfnamefont {M.}~\bibnamefont {Gu}},\ and\ \bibinfo {author} {\bibfnamefont {V.}~\bibnamefont {Vedral}},\ }\bibfield  {title} {\bibinfo {title} {Quantum processes which do not use coherence},\ }\href {https://doi.org/10.1103/PhysRevX.6.041028} {\bibfield  {journal} {\bibinfo  {journal} {Phys. Rev. X}\ }\textbf {\bibinfo {volume} {6}},\ \bibinfo {pages} {041028} (\bibinfo {year} {2016})}\BibitemShut {NoStop}%
\bibitem [{\citenamefont {Daki\ifmmode~\acute{c}\else \'{c}\fi{}}\ \emph {et~al.}(2010)\citenamefont {Daki\ifmmode~\acute{c}\else \'{c}\fi{}}, \citenamefont {Vedral},\ and\ \citenamefont {Brukner}}]{DVB10}%
  \BibitemOpen
  \bibfield  {author} {\bibinfo {author} {\bibfnamefont {B.}~\bibnamefont {Daki\ifmmode~\acute{c}\else \'{c}\fi{}}}, \bibinfo {author} {\bibfnamefont {V.}~\bibnamefont {Vedral}},\ and\ \bibinfo {author} {\bibfnamefont {C.}~\bibnamefont {Brukner}},\ }\bibfield  {title} {\bibinfo {title} {Necessary and sufficient condition for nonzero quantum discord},\ }\href {https://doi.org/10.1103/PhysRevLett.105.190502} {\bibfield  {journal} {\bibinfo  {journal} {Phys. Rev. Lett.}\ }\textbf {\bibinfo {volume} {105}},\ \bibinfo {pages} {190502} (\bibinfo {year} {2010})}\BibitemShut {NoStop}%
\bibitem [{\citenamefont {Shi}\ \emph {et~al.}(2012)\citenamefont {Shi}, \citenamefont {Sun}, \citenamefont {Jiang}, \citenamefont {Yan},\ and\ \citenamefont {Du}}]{SSJ+12}%
  \BibitemOpen
  \bibfield  {author} {\bibinfo {author} {\bibfnamefont {M.}~\bibnamefont {Shi}}, \bibinfo {author} {\bibfnamefont {C.}~\bibnamefont {Sun}}, \bibinfo {author} {\bibfnamefont {F.}~\bibnamefont {Jiang}}, \bibinfo {author} {\bibfnamefont {X.}~\bibnamefont {Yan}},\ and\ \bibinfo {author} {\bibfnamefont {J.}~\bibnamefont {Du}},\ }\bibfield  {title} {\bibinfo {title} {Optimal measurement for quantum discord of two-qubit states},\ }\href {https://doi.org/10.1103/PhysRevA.85.064104} {\bibfield  {journal} {\bibinfo  {journal} {Phys. Rev. A}\ }\textbf {\bibinfo {volume} {85}},\ \bibinfo {pages} {064104} (\bibinfo {year} {2012})}\BibitemShut {NoStop}%
\bibitem [{\citenamefont {Jevtic}\ \emph {et~al.}(2014)\citenamefont {Jevtic}, \citenamefont {Pusey}, \citenamefont {Jennings},\ and\ \citenamefont {Rudolph}}]{JPJR14}%
  \BibitemOpen
  \bibfield  {author} {\bibinfo {author} {\bibfnamefont {S.}~\bibnamefont {Jevtic}}, \bibinfo {author} {\bibfnamefont {M.}~\bibnamefont {Pusey}}, \bibinfo {author} {\bibfnamefont {D.}~\bibnamefont {Jennings}},\ and\ \bibinfo {author} {\bibfnamefont {T.}~\bibnamefont {Rudolph}},\ }\bibfield  {title} {\bibinfo {title} {Quantum steering ellipsoids},\ }\href {https://doi.org/10.1103/PhysRevLett.113.020402} {\bibfield  {journal} {\bibinfo  {journal} {Phys. Rev. Lett.}\ }\textbf {\bibinfo {volume} {113}},\ \bibinfo {pages} {020402} (\bibinfo {year} {2014})}\BibitemShut {NoStop}%
\bibitem [{\citenamefont {Cheng}\ \emph {et~al.}(2016)\citenamefont {Cheng}, \citenamefont {Milne}, \citenamefont {Hall},\ and\ \citenamefont {Wiseman}}]{CMHW16}%
  \BibitemOpen
  \bibfield  {author} {\bibinfo {author} {\bibfnamefont {S.}~\bibnamefont {Cheng}}, \bibinfo {author} {\bibfnamefont {A.}~\bibnamefont {Milne}}, \bibinfo {author} {\bibfnamefont {M.~J.~W.}\ \bibnamefont {Hall}},\ and\ \bibinfo {author} {\bibfnamefont {H.~M.}\ \bibnamefont {Wiseman}},\ }\bibfield  {title} {\bibinfo {title} {Volume monogamy of quantum steering ellipsoids for multiqubit systems},\ }\href {https://doi.org/10.1103/PhysRevA.94.042105} {\bibfield  {journal} {\bibinfo  {journal} {Phys. Rev. A}\ }\textbf {\bibinfo {volume} {94}},\ \bibinfo {pages} {042105} (\bibinfo {year} {2016})}\BibitemShut {NoStop}%
\bibitem [{\citenamefont {Peres}(1996)}]{Per96}%
  \BibitemOpen
  \bibfield  {author} {\bibinfo {author} {\bibfnamefont {A.}~\bibnamefont {Peres}},\ }\bibfield  {title} {\bibinfo {title} {Separability criterion for density matrices},\ }\href {https://doi.org/10.1103/PhysRevLett.77.1413} {\bibfield  {journal} {\bibinfo  {journal} {Phys. Rev. Lett.}\ }\textbf {\bibinfo {volume} {77}},\ \bibinfo {pages} {1413} (\bibinfo {year} {1996})}\BibitemShut {NoStop}%
\bibitem [{\citenamefont {Horodecki}\ \emph {et~al.}(1996)\citenamefont {Horodecki}, \citenamefont {Horodecki},\ and\ \citenamefont {Horodecki}}]{HHH96}%
  \BibitemOpen
  \bibfield  {author} {\bibinfo {author} {\bibfnamefont {M.}~\bibnamefont {Horodecki}}, \bibinfo {author} {\bibfnamefont {P.}~\bibnamefont {Horodecki}},\ and\ \bibinfo {author} {\bibfnamefont {R.}~\bibnamefont {Horodecki}},\ }\bibfield  {title} {\bibinfo {title} {Separability of mixed states: necessary and sufficient conditions},\ }\href {https://doi.org/https://doi.org/10.1016/S0375-9601(96)00706-2} {\bibfield  {journal} {\bibinfo  {journal} {Physics Letters A}\ }\textbf {\bibinfo {volume} {223}},\ \bibinfo {pages} {1} (\bibinfo {year} {1996})}\BibitemShut {NoStop}%
\bibitem [{\citenamefont {Wootters}(1998)}]{Woo98}%
  \BibitemOpen
  \bibfield  {author} {\bibinfo {author} {\bibfnamefont {W.~K.}\ \bibnamefont {Wootters}},\ }\bibfield  {title} {\bibinfo {title} {Entanglement of formation of an arbitrary state of two qubits},\ }\href {https://doi.org/10.1103/PhysRevLett.80.2245} {\bibfield  {journal} {\bibinfo  {journal} {Phys. Rev. Lett.}\ }\textbf {\bibinfo {volume} {80}},\ \bibinfo {pages} {2245} (\bibinfo {year} {1998})}\BibitemShut {NoStop}%
\bibitem [{\citenamefont {Shi}\ \emph {et~al.}(2011)\citenamefont {Shi}, \citenamefont {Yang}, \citenamefont {Jiang},\ and\ \citenamefont {Du}}]{SYF+11}%
  \BibitemOpen
  \bibfield  {author} {\bibinfo {author} {\bibfnamefont {M.}~\bibnamefont {Shi}}, \bibinfo {author} {\bibfnamefont {W.}~\bibnamefont {Yang}}, \bibinfo {author} {\bibfnamefont {F.}~\bibnamefont {Jiang}},\ and\ \bibinfo {author} {\bibfnamefont {J.}~\bibnamefont {Du}},\ }\bibfield  {title} {\bibinfo {title} {Quantum discord of two-qubit rank-2 states},\ }\href {https://doi.org/10.1088/1751-8113/44/41/415304} {\bibfield  {journal} {\bibinfo  {journal} {Journal of Physics A: Mathematical and Theoretical}\ }\textbf {\bibinfo {volume} {44}},\ \bibinfo {pages} {415304} (\bibinfo {year} {2011})}\BibitemShut {NoStop}%
\bibitem [{\citenamefont {Costa}\ and\ \citenamefont {Angelo}(2016)}]{CA16}%
  \BibitemOpen
  \bibfield  {author} {\bibinfo {author} {\bibfnamefont {A.~C.~S.}\ \bibnamefont {Costa}}\ and\ \bibinfo {author} {\bibfnamefont {R.~M.}\ \bibnamefont {Angelo}},\ }\bibfield  {title} {\bibinfo {title} {Quantification of einstein-podolsky-rosen steering for two-qubit states},\ }\href {https://doi.org/10.1103/PhysRevA.93.020103} {\bibfield  {journal} {\bibinfo  {journal} {Phys. Rev. A}\ }\textbf {\bibinfo {volume} {93}},\ \bibinfo {pages} {020103} (\bibinfo {year} {2016})}\BibitemShut {NoStop}%
\bibitem [{\citenamefont {Guo}\ and\ \citenamefont {Wu}(2014)}]{GW14}%
  \BibitemOpen
  \bibfield  {author} {\bibinfo {author} {\bibfnamefont {Y.}~\bibnamefont {Guo}}\ and\ \bibinfo {author} {\bibfnamefont {S.}~\bibnamefont {Wu}},\ }\bibfield  {title} {\bibinfo {title} {Quantum correlation exists in any non-product state},\ }\href {https://doi.org/10.1038/srep07179} {\bibfield  {journal} {\bibinfo  {journal} {Sci. Rep}\ }\textbf {\bibinfo {volume} {4}},\ \bibinfo {pages} {7179} (\bibinfo {year} {2014})}\BibitemShut {NoStop}%
\end{thebibliography}%


%

\end{document}